\documentclass[reqno,11pt]{amsart}
\newtheorem{definition}{Definition}
\usepackage{graphicx}
\usepackage{xcolor}
\usepackage{tikz}
\usepackage{tikz-cd}
\usepackage{caption}
\usepackage{esint}

\usepackage{adjustbox}
\usepackage{amssymb}

\usepackage{amsmath,empheq}
\usepackage{footnote}
\makesavenoteenv{minipage}	% If you want to include minipages.	
\makesavenoteenv{itemize}
\makesavenoteenv{enumerate}

\newcommand{\be}{\begin{equation}}
\newcommand{\ee}{\end{equation}}

\newcommand{\m}{\mu}
\newcommand{\w}{\tau}

\DeclareMathSymbol{\Lambda}{\mathord}{operators}{"03}

\usepackage{enumerate}% http://ctan.org/pkg/enumerate
\usepackage{enumitem}

\usepackage{appendix}

\newtheorem{thm}{Theorem}[section]
\newtheorem{prop}[thm]{Proposition}
\newtheorem{theorem}[thm]{Theorem}

\newtheorem{lemma}[thm]{Lemma}

\newtheorem{remark}[thm]{\it Remark}

\usepackage{tikz-3dplot}
\usepackage{xifthen}
\usepackage{caption}

\tdplotsetmaincoords{60}{125}
\tdplotsetrotatedcoords{0}{20}{0} %<- rotate around (z,y,z)

\newcommand{\bff}{\boldsymbol}

\usepackage[
labelfont=sf,
hypcap=false,
format=hang
]{caption}

\usepackage{cancel}
\usepackage{soul}
\usepackage[normalem]{ulem}

\begin{document}

\title[Non-Abelian elastic collisions and discrete analytic functions]{Non-Abelian elastic collisions, associated difference systems of equations  and discrete analytic functions}
\author[P. Kassotakis]{Pavlos Kassotakis}
\address{Pavlos Kassotakis, Department of Mathematical Methods in Physics, Faculty of Physics,
University of Warsaw, Pasteura 5, 02-093, Warsaw, Poland}
 \email{Pavlos.Kassotakis@fuw.edu.pl, pavlos1978@gmail.com}

\author[T. Kouloukas]{Theodoros Kouloukas}
\address{Theodoros Kouloukas, School of Computing and Digital Media, London Metropolitan University, 166-220 Holloway Rd, London N7 8DB, U.K.}
\email{t.kouloukas@londonmet.ac.uk}

\author[M. Nieszporski]{Maciej Nieszporski}
\address{Maciej Nieszporski, Department of Mathematical Methods in Physics, Faculty of Physics,
University of Warsaw, Pasteura 5, 02-093, Warsaw, Poland}
 \email{Maciej.Nieszporski@fuw.edu.pl}

\date{\today}

\begin{abstract}
We extend  the equations of motion that describe non-relativistic elastic collision of two particles in one dimension to an arbitrary associative algebra. Relativistic elastic collision equations turn out to be a
particular  case of these generic equations. Furthermore, we show that these equations can be reinterpreted as difference systems defined on the ${\mathbb Z}^2$ graph
and this reinterpretation  relates (unifies) the linear and the non-linear approach of discrete analytic functions.

\end{abstract}

\maketitle

% Uncomment for keywords
\vspace{2pc}

\setcounter{tocdepth}{2}

%\tableofcontents

%%%%%%%%%%%%%%%%%%%%%%%%
%%%%%%%%%%%%%%%%%%%%%%%%
%%%%%%%%%%%%%%%%%%%%%%%%

\section{Introduction}

The physical phenomenon of the non-relativistic elastic
collisions of particles is governed by the equations that give rise to a
map with
the Yang--Baxter property (see Section \ref{section2})  and after a suitable reinterpretation leads
to a system of linear difference equations \cite{Kouloukas-2017,KOULOUKAS:2023}.
In this article, we show that abandoning the assumption of commutativity of variables,
even the simple linear integrable difference system of equations associated with
elastic collisions, yield
nonlinear, non-abelian integrable difference system of equations. From
this non-abelian difference
system we can recover the abelian one as a special case.
Actually, in general we obtain several abelian systems with
different features. In that
respect, the non-abelian system serves as a “top” system since it
includes (unifies) the abelian
ones. In other words, if we consider
two or several systems of equations where the dependent variables are
elements of a field, that is
the multiplication of variables is abelian (commutative), then one can
try to unify these systems in the
following sense:

\begin{center}
\em First, extend one of the systems to the non-abelian setting i.e.  allow all of its  variables to be
elements of an %associative noncommutative 
 algebra $\mathcal{A}$  rather than  a field. Second, show that the original abelian systems
 are special cases of the non-abelian one, which are obtained by specifying  the algebra (or some of its subspaces) in question.
 \end{center}

In this article we present two examples of such a procedure. 
 In the first illustrative example, we unify, in the aforementioned sense, the equations that describe non-relativistic elastic collisions of point-mass particles with a system of equations that includes the equations of relativistic elastic collisions as a subsystem.
 In the second  example   (that can be considered as an extension and a reinterpretation of the first one)  we unify the linear  and the nonlinear approach of discrete analytic functions.
In the classical article on discrete analytic functions \cite{BoMeSu},
the authors show that
Duffin's linear theory of discrete analytic functions \cite{Duffin2}
can be viewed as a linearisation of the
nonlinear theory that
``is based on the ideas by Thurston
\cite{Thurston} and declares circle patterns to be natural discrete analogs of analytic functions".
Here instead we present the two theories as special abelian cases of the non-abelian theory.

There is not a unique procedure to extend an abelian system to its non-abelian counterpart. To resolve this non-uniqueness issue, we perform the non-abelian extension on the level of the underlying Lax formulation known from the theory of integrable systems \cite{Lax}. The procedure of using the Lax pair formulation to lift an abelian system to its non-abelian version has already been considered in  the literature, see e.g. \cite{MIKHAILOV:2024}.  

  The following system of non-abelian equations serves as a  central point of this article
\begin{subequations}\label{TheSystemI}
%\begin{empheq}[left=\empheqlbrace]{align}
\begin{gather}
\label{icentr:2:1}
 v^{i}{'} -v^{j}{'} =v^{j}-v^{i},\\
\label{icentr:2:2}
\m^{i}{'} +\m^{j}{'} =\m^{j}+\m^{i},\\
\label{icentr:2:3}
\m^{i}{'} \m^{j}{'} =\m^{j} \m^{i},\\
\label{icentr:2:4}
\m^{i}{'} v^{j}{'} +v^{i}{'} \m^{j}{'}=\m^{j} v^{i}+v^{j} \m^{i}.
\end{gather}
%&\end{empheq}
\end{subequations}
where $i\neq j\in\{1,\ldots,N\},$ $N\geq 2\in \mathbb{N}$ and we assume
 that both the primed 
variables $\m^{i}{'}$, $v^{i}{'}$ and the non-primed ones $\m^{i}$, $v^{i}$
belong to a unital associative algebra $\mathcal{A}$  over a field $\mathbb{F}$. 

System (\ref{TheSystemI})  is equivalent  to the matrix refactorization problem, 
\begin{align} \label{Lax-fI}
L(v^{i}{'},\m^{i}{'};\lambda)L(v^{j}{'},\m^{j{'}};\lambda)=L(v^{j},\m^{j};\lambda)L(v^{i},\m^{i};\lambda)\;,
\end{align}
where $L$ is the following Lax matrix
\begin{align} \label{Lax_mat}
L(v^{i},\m^{i};\lambda):=\begin{pmatrix}
\m^{i}+\lambda& \lambda v^{i}\\
0&\m^{i}-\lambda
\end{pmatrix},
\end{align} 
and $\lambda$ stands for a spectral parameter,  which  is assumed to belong to the  center of the algebra. The refactorization problem (\ref{Lax-fI}) in components reads

\begin{subequations} \label{cc_ci}
    \begin{align} \label{cc_c1i}
    (\m^i{'}+\lambda)(\m^j{'}+\lambda)=&(\m^j+\lambda)(\m^i+\lambda),\\ \label{cc_c2i}
    (\m^i{'}-\lambda)(\m^j{'}-\lambda)=&(\m^j-\lambda)(\m^i-\lambda),\\ \label{cc_c3i}
    (\m^i{'}+\lambda)v^j{'}+v^i{'} (\m^j{'}-\lambda)=&(\m^j+\lambda)v^i+v^j(\m^i-\lambda).
    \end{align}
    \end{subequations}
Demanding that  (\ref{cc_ci})  hold for every $\lambda,$  (\ref{cc_c1i}) is equivalent to (\ref{cc_c2i}) and together with (\ref{cc_c3i}) we obtain exactly
the system of non-abelian equations (\ref{TheSystemI}). 
The sub-system, (\ref{icentr:2:2}),(\ref{icentr:2:3}) could be re-interpreted as  a difference system in edge variables associated  with  {\em the discrete nonlinear chiral field  (or  $\sigma$-model) equation} \cite{Cherednik_1982} (see Section \ref{re0}). Various reductions of (\ref{icentr:2:2}),(\ref{icentr:2:3}) have been considered in  the literature. In detail, in \cite{Kouloukas_2009,Kouloukas_2011} where the algebra $\mathcal{A}$ was considered as the algebra of the $n\times n$ matrices over $\mathbb{C}$, the solution of (\ref{icentr:2:2}),(\ref{icentr:2:3}) with respect to $\m^{i}{'},\m^{j}{'}$ has been explicitly  presented as a multi-component Yang-Baxter map. 
The  Lax matrix (\ref{Lax_mat}) with entries in a commutative ring and in connection with the linear theory of discrete analytic functions can be found  in \cite{Bobenko:2008}.

The article is organized as follows.

In Section \ref{sec-laws}, we consider  two  choices of the underlying  algebra $\mathcal{A}$. 
The first choice, is  to take the field of real numbers $\mathbb{R}$ as the algebra $\mathcal{A}.$ This choice
 leads to the equations that describe the elastic collisions between the $i$-th and the $j$-th non-relativistic particles, where $m^i:=\frac{1}{\mu^i}$, $m^j:=\frac{1}{\mu^j}$ and $v^i$, $v^j$ respectively denote the
masses and the velocities of the  particles before collision, while the primed variables denote
these quantities after collision.
The second choice is to consider all $\mu$'s as off-diagonal $2 \times 2$ matrices over the field
$\mathbb{R}$ and all $v$'s as diagonal $2 \times 2$ matrices over the same field. This choice results in the equations that describe the elastic collisions of relativistic particles. 

In Section \ref{sec-maps}, first we  assume that $\mathcal{A}$ is  a division ring and then we show that the system (\ref{TheSystemI}) can be solved for four  (out of eight variables) and give rise  
to maps $\mathcal{A}^2 \times \mathcal{A}^2 \to \mathcal{A}^2 \times \mathcal{A}^2.$ 
 We  prove that  these maps are multidimensional compatible, while their {\em companion} maps satisfy the Yang-Baxter property.

 In  Section \ref{re0},  we reinterpret    (\ref{TheSystemI}) as a system of difference equations defined on edges of  the $\mathbb{Z}^N$ graph
and furthermore we associate  difference systems of  equations defined on vertices of the same graph. Specifically, the system of vertex equations on the functions $\chi , \, \psi ,\, \omega , \, \sigma , \, \phi  \, : {\mathbb Z}^N \to \mathcal{ A}$, where  $\mathcal{ A}$  a division ring, reads
%\begin{subequations}\label{idolonii}
\begin{gather}\label{ido_i}
\begin{gathered}
\psi_i-\psi=\phi_i\phi^{-1},\\
\sigma_i-\sigma=\phi^{-1}_i\phi,\\
\chi_i+\chi=\phi_i(\omega_i-\omega)\phi^{-1},
\end{gathered}
\end{gather}
%\end{subequations}
where $i\in\{1,\ldots, N\}$ and the subscripts denote forward shifts  in the $i$-th direction (see Figure \ref{fig1} and Section \ref{re} for details). 
In particular, having eliminated $\psi,\omega,$ and $\sigma$  from  (\ref{ido_i}), we obtain  {\em the discrete nonlinear chiral field  (or  $\sigma$-model) equation} \cite{Cherednik_1982} for the function $\phi$ that reads
\begin{align} \label{vvv01}
\begin{aligned}
\phi_{ij}\left(\phi_j^{-1}-\phi_i^{-1}\right)=\left(\phi_i-\phi_j\right)\phi^{-1},
\end{aligned}
\end{align}
 which is coupled to the following linear equation for the function $\chi$
\begin{align} \label{vvv02}
\begin{aligned}
\phi_{ij}\phi_j^{-1}(\chi_j+\chi)+(\chi_{ij}+\chi_j)\phi_j\phi^{-1}=\phi_{ij}\phi_i^{-1}(\chi_i+\chi)+(\chi_{ij}+\chi_i)\phi_i\phi^{-1}.
\end{aligned}
\end{align}

In Section \ref{sec4n}, we show that the solutions of the  difference systems of  equations defined on vertices of quads  on the $\mathbb{Z}^N$ graph, satisfy systems of three-dimensional vertex equations. For example, we show that the solutions of (\ref{vvv01}) and (\ref{vvv02}) satisfy the following three-dimensional system
\begin{equation} \label{a3deq}
\begin{array}{c}
 \phi_{ij}\phi_{ik}^{-1}\phi_{jk}=\phi_{jk}\phi_{ik}^{-1}\phi_{ij},\\
    \left(\phi_{ik}^{-1}\left(\chi_{ik}+\chi_k\right)-\phi_{jk}^{-1}\left(\chi_{jk}+\chi_k\right)\right)\phi_k+
    \left(\phi_{ij}^{-1}\left(\chi_{ij}+\chi_i\right)-\phi_{ik}^{-1}\left(\chi_{ik}+\chi_i\right)\right)\phi_i\\
    +\left(\phi_{jk}^{-1}\left(\chi_{jk}+\chi_j\right)-\phi_{ij}^{-1}\left(\chi_{ij}+\chi_j\right)\right)\phi_j=0,
    \end{array}
\end{equation}     
where $i\neq j\neq k \in \{1,\ldots, N\}.$

In the final Section \ref{secf}, we show how the linear  and the nonlinear  approach of discrete analytic functions are related via system (\ref{ido_i}).  When $\mathcal{ A}=\mathbb{C}$, system (\ref{ido_i})  can be treated as the basic equations  of the theory of discrete analytic functions \cite{Duffin2}. While, when as $\mathcal{ A}$ we choose the subspace of $2\times 2$ off-diagonal matrices over a commutative ring,
 system (\ref{ido_i}) can be reinterpreted as the  nonlinear approach to discrete analytic functions \cite{Thurston}.
  Therefore, we relate
 (unify) in the sense mentioned in the beginning of this introduction, the linear and the nonlinear approach of discrete analytic functions.

\section{ Reductions to classical and relativistic elastic collisions of point-mass particles} \label{sec-laws}

In this Section we show that in the abelian case, system  (\ref{TheSystemI})  describes the elastic collisions of non-relativistic particles. Furthermore we show that the relativistic elastic collisions, with an additional degree of freedom, are obtained by a suitable choice of the underlying  algebra $\mathcal{A}$.
\subsection{Classical head-on elastic collisions of  point-mass particles}
Let the point-mass particles labeled by the index $i\in\mathbb{N},$ with  initial respective masses and velocities $m^{i}$ and $v^{i},$  undergo pairwise one-dimensional elastic collisions. The conservation of kinetic energy and the conservation of momentum  in the event of collision of two particles respectively read
\begin{align} \label{E+M}
\begin{aligned}
m^{i}\left(v^{i}\right)^2+m^{j}\left(v^{j}\right)^2=m^{i}{'}\left(v^{i}{'}\right)^2+m^{j}{'}\left(v^{j}{'}\right)^2,\\
m^{i}v^{i}+m^{j}v^{j}=m^{i}{'}v^{i}{'}+m^{j}{'}v^{j}{'},
\end{aligned} && i\neq j \in\mathbb{N},
\end{align}
where $m^{i}{'},m^{j}{'}$ and $v^{i}{'},v^{j}{'},$ denote the respective masses and velocities after the collision.
Demanding that momentum is  preserved  in any  Galilean frame leads to the conservation of mass 
\begin{equation} \label{conserv_sum_mass}
m^i +m^j  =m^i{'} +m^j{'}.
\end{equation}
In the standard   elastic collision process,  it is assumed that there is no transfer of mass between the particles, i.e. we have
\begin{equation} \label{conserv_mass}
m^i{'}=m^i, \qquad m^j{'}=m^j. \end{equation}
 Then (\ref{E+M}) equivalently reads
\begin{align*}
m^{i} \left(v^{i}+v^{i}{'}\right)\left(v^{i}-v^{i}{'}\right)+m^{j} \left(v^{j}+v^{j}{'}\right)\left(v^{j}-v^{j}{'}\right)=0,\\
m^{i} \left(v^{i}-v^{i}{'}\right)+m^{j} \left(v^{j}-v^{j}{'}\right)=0,
\end{align*}
or when $m^{i} \left(v^{i}-v^{i}{'}\right)\neq 0,$
\begin{align} \label{2:1}
\begin{aligned}
v^{i}+v^{i}{'}=v^{j}+v^{j}{'},\\ %\label{2:2}
m^{i} \left(v^{i}-v^{i}{'}\right)+m^{j} \left(v^{j}-v^{j}{'}\right)=0.
\end{aligned}
\end{align}

%%%%%%%%%%%%%%%%%%%%%%%%%%%%%%%%%%%%%%%%%%%%%%%%%%%%%%%%%%%%%%%%%%%%%%%%%%%%%%%%%
In the following Proposition we show how to obtain the equations that describe classical elastic collisions of point-mass particles from an appropriate reduction of system (\ref{TheSystemI}).

\begin{prop} 
When we consider the underlying algebra $\mathcal{A}$ to be abelian,  e.g. $\mathcal{A}=\mathbb{R}$, the system of equations (\ref{TheSystemI}), for $\mu^i{'}=\mu^i=(m^i)^{-1}$ and $\mu^j{'}=\mu^j=(m^j)^{-1},$ reduces to
%\begin{subequations}\label{TheSystem1_22}
%\begin{empheq}[left=\empheqlbrace]{align}
\begin{align}\label{TheSystem1_22}
m^{i}{'} =& m^{j}, &
m^{j}{'} =& m^{i}, &
 v^{i}{'} -v^{j}{'} =&v^{j}-v^{i},&
m^{i} (v^{j}{'} - v^{j}) =&m^{j} (v^{i}-v^{i}{'}),
\end{align}
%&\end{empheq}
%\end{subequations}
or to 
%\begin{subequations}\label{TheSystem1}
%\begin{empheq}[left=\empheqlbrace]{align}
\begin{align} \label{TheSystem1}
m^{i}{'} =& m^{i},&
m^{j}{'} =& m^{j}, &
 v^{i}{'} -v^{j}{'} =&v^{j}-v^{i},&
m^{i} (v^{j}{'} - v^{j}) =&m^{j} (v^{i}-v^{i}{'}).
\end{align}
%&\end{empheq}
%\end{subequations}
The latter  coincides with (\ref{conserv_mass}),(\ref{2:1}), that is the equations of classical elastic collisions of two point-mass particles with masses $m^i,m^j>0,$ with velocities $v^i,v^j$ before the collision and velocities $v^i{'},v^j{'}$ after the collision.
\end{prop} 

\begin{proof} 
For $\mathcal{A}=\mathbb{R},$ equations
(\ref{icentr:2:2}) and (\ref{icentr:2:3}) yield $\mu^i{'}=\mu^j$ and $\mu^j{'}=\mu^i$ or
$\mu^i{'}=\mu^i=:(m^i)^{-1}$ and $\mu^j{'}=\mu^j=:(m^j)^{-1}.$ 
In the latter case  system (\ref{TheSystemI}) becomes exactly  (\ref{TheSystem1})
that coincides with (\ref{conserv_mass}) and (\ref{2:1}).
\end{proof}

\subsection{Head-on elastic collisions of relativistic point-mass particles}

Let the relativistic point-mass particles labeled by the index $i\in\mathbb{N},$ with  initial respective  rest masses $m_0^i$ and $m_0^j$ and respective velocities $u^i$ and $u^j$ before the collision and velocities $u^i{'},u^j{'}$ after the collision.
The conservation of relativistic momentum and energy under the elastic collision respectively read
\begin{align} \label{el_col_0}
\begin{aligned}
\frac{m_{0}^{i}u^{i}{'}}{\sqrt{1-\left(\frac{u^{i}{'}}{c}\right)^2}}+\frac{m_{0}^{j}u^{j}{'}}{\sqrt{1-\left(\frac{u^{j}{'}}{c}\right)^2}}&=
\frac{m_{0}^{i}u^{i}}{\sqrt{1-\left(\frac{u^{i}}{c}\right)^2}}+\frac{m_{0}^{j}u^{j}}{\sqrt{1-\left(\frac{u^{j}}{c}\right)^2}},\\
\frac{m_{0}^{i}}{\sqrt{1-\left(\frac{u^{i}{'}}{c}\right)^2}}+\frac{m_{0}^{j}}{\sqrt{1-\left(\frac{u^{j}{'}}{c}\right)^2}}&=
\frac{m_{0}^{i}}{\sqrt{1-\left(\frac{u^{i}}{c}\right)^2}}+\frac{m_{0}^{j}}{\sqrt{1-\left(\frac{u^{j}}{c}\right)^2}},
\end{aligned}
\end{align}
where $c$ the speed of light. Introducing the change of variables 
 \begin{align*}
 \begin{aligned}
 \mathbb{R}^{+}\ni x^i\mapsto u^i=& c\frac{\left(x^i\right)^2-1}{\left(x^i\right)^2+1}\in (-c,c),\\
 \mathbb{R}^{+}\ni x^i{'}\mapsto u^i{'}=& c\frac{\left(x^i{'}\right)^2-1}{\left(x^i\right)^2+1}\in (-c,c),
   \end{aligned}&& i=1,\ldots, N,
   \end{align*}
    that is a bijection,  equations (\ref{el_col_0}) become
\begin{align*}
m_{0}^{i}\left(x^i{'}-\frac{1}{x^i{'}}\right)+m_{0}^{j}\left(x^j{'}-\frac{1}{x^j{'}}\right)&=
m_{0}^{i}\left(x^i-\frac{1}{x^i}\right)+m_{0}^{j}\left(x^j-\frac{1}{x^j}\right),\\
m_{0}^{i}\left(x^i{'}+\frac{1}{x^i{'}}\right)+m_{0}^{j}\left(x^j{'}+\frac{1}{x^j{'}}\right)&=
m_{0}^{i}\left(x^i+\frac{1}{x^i}\right)+m_{0}^{j}\left(x^j+\frac{1}{x^j}\right).
\end{align*}
By adding and subtracting the equations above we obtain the conservation laws
\begin{align}\label{re-m}
m_{0}^{i}\left(x^i{'}-x^i\right)&=m_{0}^{j}\left(x^j-x^j{'}\right),\\ \label{re-e}
m_{0}^{i}\left(\frac{1}{x^i{'}}-\frac{1}{x^i}\right)&=
m_{0}^{j}\left(\frac{1}{x^j}-\frac{1}{x^j{'}}\right),
\end{align}
which  are equivalent to the  conservation of relativistic momentum and energy.

 It turns out that the  system of equations (\ref{TheSystemI}) under a suitable choice of the underlying  algebra $\mathcal{A}$  includes the relativistic elastic collisions, henceforth it relates the Newtonian and the relativistic approaches.
\begin{lemma}\label{lema000}
Let $\mathcal{A}=\mathcal{A}_o\bigoplus \mathcal{A}_e$  be a $\mathbb{Z}_2-$graded  algebra  over $\mathbb{C}$. Let  the odd subspace $\mathcal{A}_o$ of the algebra $\mathcal{A}$ be spanned by the $2\times 2$ off-diagonal matrices
\begin{align*}
 \m^i=\begin{pmatrix}
0&x^i\\
X^i&0
\end{pmatrix},\qquad \m^i{'}=\begin{pmatrix}
0&x^i{'}\\
X^i{'}&0
\end{pmatrix},
\end{align*}
and the even subspace $\mathcal{A}_e$ be spanned by the $2\times 2$ diagonal matrices
\begin{align*}v^i=\begin{pmatrix}
y^i&0\\
0&Y^i
\end{pmatrix},\qquad v^i{'}=\begin{pmatrix}
y^i{'}&0\\
0&Y^i{'}
\end{pmatrix}.
\end{align*} 
Then the system of equations (\ref{icentr:2:1})-(\ref{icentr:2:4}), restricted on these subspaces respectively reads
\begin{subequations}\label{red_0_2}
\begin{align}\label{red_01}
  y^i{'}-y^j{'}=& y^j-y^i, & Y^i{'}-Y^j{'}=& Y^j-Y^i, \\ \label{red_02}
  x^i{'}+x^j{'}=& x^i+x^j, & X^i{'}+X^j{'}=& X^i+X^j, \\ \label{red_03}
  x^i{'}X^j{'} =& X^ix^j, & X^i{'}x^j{'} =& x^iX^j, \\ \label{red_04}
  x^j{'}y^i{'}+x^i{'}Y^j{'} &=x^iy^j+x^jY^i,&
  X^j{'}Y^i{'}+X^i{'}y^j{'} &=X^iY^j+X^jy^i,
\end{align}
\end{subequations}
and it satisfies the following invariant relations
\begin{align}\label{inv_re_0}
x^i{'}X^i{'}=&x^iX^i,& x^j{'}X^j{'}=&x^jX^j,& y^i{'}-Y^i{'}=&Y^i-y^i,& y^j{'}-Y^j{'}=&Y^j-y^j.
\end{align}
\end{lemma}

\begin{proof}
By  restricting (\ref{TheSystemI}) on the odd and the even subspaces of the $\mathbb{Z}_2-$graded algebra $\mathcal{A}$, we obtain exactly (\ref{red_0_2}).

The system of equations (\ref{red_0_2}) admits two solutions with respect to the primed variables which they read
\begin{align*}
  x^i{'}= & x^j,& X^i{'}= & X^j,& x^j{'}= & x^i,& X^j{'}= & X^i,\\
  y^i{'}= & y^j,& Y^i{'}= & Y^j,& y^j{'}= & y^i,& Y^j{'}= & Y^i,
\end{align*}
or
\begin{subequations}\label{red_2_2}
\begin{align}\label{red_1}
x^i{'}=&X^i\frac{x^i+x^j}{X^i+X^j},&X^i{'}=&x^i\frac{X^i+X^j}{x^i+x^j},\\ \label{red_2}
x^j{'}=&X^j\frac{x^i+x^j}{X^i+X^j},&X^j{'}=&x^j\frac{X^i+X^j}{x^i+x^j},\\ \label{red_3}
y^i{'}=&X^i\frac{Y^j-y^i}{X^i+X^j}+\frac{x^jY^i+x^iy^j}{x^i+x^j},&Y^i{'}=&x^i\frac{y^j-Y^i}{x^i+x^j}+\frac{X^jy^i+X^iY^j}{X^i+X^j},\\ \label{red_4}
y^j{'}=&x^j\frac{Y^i-y^j}{x^i+x^j}+\frac{X^jy^i+X^iY^j}{X^i+X^j},&Y^j{'}=&X^j\frac{y^i-Y^j}{X^i+X^j}+\frac{x^jY^i+x^iy^j}{x^i+x^j}.
\end{align}
\end{subequations}
By substituting (\ref{red_2_2}) into (\ref{inv_re_0}), we prove that the invariant relations hold. 
\end{proof}

\begin{prop} \label{prop_rel}
The system of equations (\ref{red_02}),(\ref{red_03}) is equivalent to (\ref{re-m}),(\ref{re-e}), hence it describes elastic collisions of relativistic pairs of particles.
\end{prop}

\begin{proof}
From Lemma \ref{lema000} we have that  (\ref{red_02}),(\ref{red_03}) respect the invariant conditions
\begin{align*}
x^i{'}X^i{'}=&x^iX^i:=\alpha^i,& x^j{'}X^j{'}=&x^jX^j:=\alpha^j.
\end{align*}
Under the change of variables $(x^i,X^i,x^j,X^j)\mapsto (x^i,\alpha^i,x^j,\alpha^j),$ the invariant conditions read
\begin{align*}
\alpha^i{'}=&\alpha^i,&\alpha^j{'}=&\alpha^j,
\end{align*}
while
(\ref{red_02}),(\ref{red_03}) become
\begin{align}
\begin{aligned}\label{pofi}
x^i{'}-x^i&=x^j-x^j{'},\\ 
\alpha^i\left(\frac{1}{x^i{'}}-\frac{1}{x^i}\right)&=\alpha^j\left(\frac{1}{x^j}-\frac{1}{x^j{'}}\right),\\ 
 \alpha^jx^i{'} x^i&=\alpha^ix^j x^j{'}.
 \end{aligned}
\end{align}
Under the re-scaling $(x^i,x^j)\mapsto \left(\sqrt{\alpha^i}x^i,\sqrt{\alpha^j}x^j\right)$  followed by the re-parametrization $\alpha^i=(m^i_0)^2,$ $\alpha^j=(m^j_0)^2,$ where $m^i_0,m^j_0$ represent the rest masses of the particles,  
(\ref{pofi}) reads
\begin{align}\label{Re_1}
m^i_0(x^i{'}-x^i)&=m^j_0(x^j-x^j{'}),\\ \label{Re_2}
m^i_0\left(\frac{1}{x^i{'}}-\frac{1}{x^i}\right)&=m^j_0\left(\frac{1}{x^j}-\frac{1}{x^j{'}}\right),\\ \label{Re_3}
 x^i{'} x^i&=x^j x^j{'}.
\end{align}
Note that from the three equations above, as expected,  only two of them  are functionally independent. Indeed, by multiplying (\ref{Re_3}) with (\ref{Re_1}) we obtain (\ref{Re_2}). It is clear that 
 equations (\ref{Re_1}),(\ref{Re_2}) coincide with (\ref{re-m}),(\ref{re-e}) and that completes the proof.

\end{proof}

\begin{remark}
In Proposition \ref{prop_rel} we proved that the sub-system of equations (\ref{red_02}),(\ref{red_03}) is equivalent to the description of elastic collisions of relativistic pairs of particles. At this point it is not clear to us if there is an underlying physical phenomenon that the whole system of equations (\ref{red_01})-(\ref{red_04}) describes. We would like here to mention that using (\ref{inv_re_0}) we re-write the equations (\ref{red_01})-(\ref{red_04}) (or (\ref{red_2_2}) in the new variables 
\begin{align*}
(x^i,X^i,x^j,X^j,y^i,Y^i,y^j,Y^j)\mapsto (x^i,(m^i_0)^2,x^j,(m^j_0)^2,y^i,p^i,y^j,p^j),
\end{align*}
where $(m^i_0)^2:=x^iX^i,$ $(m^j_0)^2:=x^jX^j,$ $p^i:=Y^i-y^i,$ $p^j:=Y^j-y^j,$
to obtain
\begin{align}\label{ext_h3a}
\begin{aligned}
 x^i{'}=& x^j\frac{m^i_0x^i+m^j_0x^j}{m^j_0x^i+m^i_0x^j},&
 x^j{'}=&x^i\frac{m^i_0x^i+m^j_0x^j}{m^j_0x^i+m^i_0x^j}, \\
  m^i_0{'}=&m^i_0, & m^j_0{'}=&m^j_0,\\
  y^i{'}=&y^j+A,&y^j{'}=&y^i+A,\\
  p^i{'}=&-p^i,&p^j{'}=&-p^j
  \end{aligned}
\end{align}
   where 
   \begin{align*}
   A:=-\left((m^i_0)^2-(m^j_0)^2\right)\frac{x^ix^j(y^i-y^j)}{(m^i_0x^i+m^j_0x^j)(m^j_0x^i+m^i_0x^j)}+x^j\left(\frac{m^j_0p^i}{m^i_0x^i+m^j_0x^j}
  +\frac{m^i_0p^j}{m^j_0x^i+m^i_0x^j}\right).
   \end{align*}
Note that the variable $y$ could be considered as an additional degree of freedom.
\end{remark}

\begin{remark}
An equivalent description of the relativistic elastic collisions of particles  is obtained by   
\begin{align*}
\m^i&=\begin{pmatrix}
0&z^{i,0}+z^{i,1}\\
z^{i,0}-z^{i,1}&0
\end{pmatrix}, &v^i&=v^i{'}=v^j=v^j{'}=:c_v\;(\mbox{constant}),
\end{align*}
where $z^{i,0},z^{i,1}\in\mathbb{C}.$  Then (\ref{TheSystemI}) reads
\begin{align} \label{relcol}
\mathbf{z^{i}{'}}=&\mathbf{z^j}+k(\mathbf{z^i},\mathbf{z^j})(\mathbf{z^i}+\mathbf{z^j}),&
\mathbf{z^{j}{'}}=&\mathbf{z^i}-k(\mathbf{z^i},\mathbf{z^j})(\mathbf{z^i}+\mathbf{z^j})\;, 
\end{align}
where $\mathbf{z^i}=(z^{i,0},z^{i,1})$, $\mathbf{z^i{'}}=(z^{i,0}{'},z^{i,1}{'})$, 
\begin{equation*} \label{K}
k(\mathbf{z^i},\mathbf{z^j})=\frac{\langle\mathbf{z^i},\mathbf{z^j}\rangle-\langle\mathbf{z^j},\mathbf{z^j}\rangle}{\langle\mathbf{z^i}+\mathbf{z^j},\mathbf{z^i}
+\mathbf{z^j}\rangle}\;,
\end{equation*}
and $\langle \ , \ \rangle$ denotes the bilinear form  $\langle\mathbf{z^i},\mathbf{z^j}\rangle:=z^{i,0} z^{j,0}-z^{i,1} z^{j,1}$. 
The map 
$\mathbf{R}:(\mathbf{z^i},\mathbf{z^j})\mapsto (\mathbf{z^i{'}},\mathbf{z^j{'}})$ represents the transformation of the momentum-energy vectors of two particles under relativistic collision \cite{KOULOUKAS:2023}. In this setting, the vectors $\mathbf{z^{i}}$ and $\mathbf{z^{i}{'}}$ correspond to the momentum-energy vectors of the two particles before and after collision, that is  
$$  
\mathbf{z^{i}}=(z^{i,0},z^{i,1}):=\left(\frac{E^i}{c},p^i\right),\quad   \mathbf{z^{i}{'}}=(z^{i,0}i{'},z^{i,1}{'}):=\left(\frac{E^i{'}}{c},p^i{'}\right), \; i\in \mathbb{N},$$
where $E^i,E^i{'}$ denote the relativistic energy of the two particles and $p^i,p^i{'}$ their momenta before and after collision respectively. 
\end{remark}

\section{Elastic collisions as maps} \label{sec-maps}

In the abelian, case system (\ref{TheSystemI})  
can be solved uniquely
with respect to the primed variables $v^{i}{'}$, $v^{j}{'}$  $\m^{i}{'}$, $\m^{j}{'}$  in terms of the un-primed ones. Provided that  $m^{i}+m^{j}\neq 0$ we obtain the maps
\begin{align} \label{YB_EC}
R_{ij}:(v^{i},m^{i}; v^{j},m^{j})\mapsto & \left(v^{i}{'},m^{i}{'}; v^{j}{'},m^{j}{'}\right),
\end{align}
where
\begin{align*} 
\begin{aligned}
v^{i}{'}=v^{j}+\frac{m^{i}-m^{j}}{m^{i}+m^{j}}\left(v^{i}-v^{j}\right), \qquad m^{i}{'}=m^{i},\\
v^{j}{'}=v^{i}+\frac{m^{i}-m^{j}}{m^{i}+m^{j}}\left(v^{i}-v^{j}\right), \qquad m^{j}{'}=m^{j}.
\end{aligned}
\end{align*}
Maps $R_{ij}$ will be referred to as {\em elastic collision maps}.
 The elastic collision  maps have the Yang-Baxter property, that is if we take three particles with given velocities $v^1$, $v^2$ and $v^3,$
after the interaction of the $1st$ particle with the $2nd$, followed by the interaction of the $1st$ with the $3rd$ and finally of the $2nd$ particle with  the $3rd$, the outgoing velocities   are the same if the order of the three interactions is reversed (see Figure \ref{fig00}).

%\input{fig000.tex}

%\tdplotsetmaincoords{75}{50} % to reset previous setting
\tdplotsetmaincoords{0}{50}
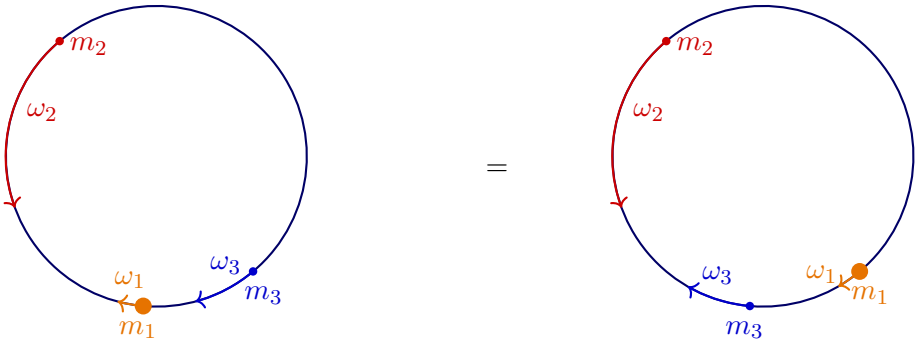
\begin{figure}[htb]\label{fig00}
\begin{center}\hspace{1.6cm}
\begin{minipage}[htb]{0.4\textwidth} \tikzcdset{every label/.append style = {font = \large}}
%\tdplotsetmaincoords{0}{50}
\begin{tikzpicture}[scale=2,tdplot_main_coords,rotate around x=90]

  % variables
  \def\rvec{0.0}
  \def\thetavec{40}
  \def\phivec{70}
  \def\R{1.0}
  \def\w{0.0}

  % axes
  \coordinate (O) at (0,0,0);
  %\draw[thick,->] (0,0,0) -- (1,0,0) node[below left]{$x$};
  %\draw[thick,->] (0,0,0) -- (0,1,0) node[below right]{$y$};
  %\draw[thick,->] (0,0,0) -- (0,0,1) node[below right]{$z$};
  %\tdplotsetcoord{P}{\rvec}{\thetavec}{\phivec}

  % circle - LHC
  \tdplotdrawarc[thick,rotate around x=90,black!60!blue]{(\R,0,0)}{\R}{0}{360}{}{}

  % nodes
  \fill[radius=0.8pt,black!20!red]     (O) circle node[left=4pt,below=2pt,right=4pt] {$m_2$};
  \fill[radius=0.8pt,black!20!blue]     (2*\R,0,0) circle    node[right=4pt,below=2pt] {$m_3$};
  \fill[radius=1.6pt,black!10!orange]     ({\R*sqrt(2)/2+\R},0,{ \R*sqrt(2)/2}) circle     node[left=2pt,below=2pt,scale=1.0] {$m_1$};

%\tdplotdrawarc[->,thick,rotate around x=89,red]{({\R*sqrt(2)/2+\R},0,{ \R*sqrt(2)/2})}{\R}{0}{90}{}{}

%arcs
\tdplotdrawarc[->,thick,rotate around x=89,black!20!blue]{ (\R,0,0)}{\R}{0}{25}{above=2pt}{$\omega_3$};

\tdplotdrawarc[->,thick,rotate around x=89,black!10!orange]{ (\R,0,0)}{\R}{45}{55}{above=2pt}{$\omega_1$};

\tdplotdrawarc[->,thick,rotate around x=89,black!20!red]{ (\R,0,0)}{\R}{180}{110}{above=2pt,right=2pt}{$\omega_2$};
\end{tikzpicture}
\end{minipage}  =  \hspace{0.9cm}
\begin{minipage}[htb]{0.4\textwidth} \tikzcdset{every label/.append style = {font = \large}}
%\tdplotsetmaincoords{0}{50}
\begin{tikzpicture}[scale=2,tdplot_main_coords,rotate around x=90]

  % variables
  \def\rvec{0.0}
  \def\thetavec{40}
  \def\phivec{70}
  \def\R{1.0}
  \def\w{0.0}

  % axes
  \coordinate (O) at (0,0,0);
  %\draw[thick,->] (0,0,0) -- (1,0,0) node[below left]{$x$};
  %\draw[thick,->] (0,0,0) -- (0,1,0) node[below right]{$y$};
  %\draw[thick,->] (0,0,0) -- (0,0,1) node[below right]{$z$};
  %\tdplotsetcoord{P}{\rvec}{\thetavec}{\phivec}

  % circle - LHC
  \tdplotdrawarc[thick,rotate around x=90,black!60!blue]{(\R,0,0)}{\R}{0}{360}{}{}

  % nodes
  \fill[radius=0.8pt,black!20!red]     (O) circle node[left=4pt,below=2pt,right=4pt] {$m_2$};
  \fill[radius=1.6pt,black!10!orange]     (2*\R,0,0) circle    node[right=4pt,below=2pt] {$m_1$};
  \fill[radius=0.8pt,black!20!blue]     ({\R*sqrt(2)/2+\R},0,{ \R*sqrt(2)/2}) circle     node[left=2pt,below=2pt,scale=1.0] {$m_3$};

%\tdplotdrawarc[->,thick,rotate around x=89,red]{({\R*sqrt(2)/2+\R},0,{ \R*sqrt(2)/2})}{\R}{0}{90}{}{}

%arcs
\tdplotdrawarc[->,thick,rotate around x=89,black!10!orange]{ (\R,0,0)}{\R}{0}{10}{above=2pt,left=0.5pt}{$\omega_1$};
\tdplotdrawarc[->,thick,rotate around x=89,black!20!blue]{ (\R,0,0)}{\R}{45}{70}{above=2pt}{$\omega_3$};
\tdplotdrawarc[->,thick,rotate around x=89,black!20!red]{ (\R,0,0)}{\R}{180}{110}{above=2pt,right=2pt}{$\omega_2$};

\end{tikzpicture}

\end{minipage}
\end{center}
\caption{The Yang-Baxter relation, realised as elastic collision of three particles moving on a circle. That is the outgoing velocities after the interaction of the 1st particle with the 2nd, followed by the interaction of the 1st with the 3rd and finally of the 2nd particle with  the 3rd (left figure), are the same if the order of the three interactions is reversed (right figure).
}\label{fig00}
\end{figure}

Alternatively, the system (\ref{TheSystemI}) can be solved for $v^{i}{'}$, $m^{i}{'}$ $v^{j}$, $m^{j}$ in terms of $v^{i}$, $m^{i}$ $v^{j}{'}$, $m^{j}{'}$ provided that  $m^{i}-m^{j}\neq 0$ to obtain
\begin{align} \label{CM_EC}
Q_{ij}:(v^{i},m^{i}; v^{j}{'},m^{j}{'})\mapsto & \left(v^{i}{'},m^{i}{'}; v^{j},m^{j}\right),
\end{align}
where
\begin{align*} 
\begin{aligned}
v^{i}{'}=v^{j}{'}+\frac{m^{i}+m^{j}{'}}{m^{i}-m^{j}{'}}\left(v^{i}-v^{j}{'}\right),&&m^{i}{'}=m^{i},\\
v^{j}=v^{i}+\frac{m^{i}+m^{j}{'}}{m^{i}-m^{j}{'}}\left(v^{i}-v^{j}{'}\right),&&m^{j}=m^{j}{'},
\end{aligned}
\end{align*}
Maps $Q_{ij},$ will be called {\em companion  elastic collision maps}.

Note that in the non-abelian case the problem of expressing some variables of the system (\ref{TheSystemI}) as  functions (maps) of the remaining ones  is more subtle.
In detail, in order to avoid indeterminacies we have to consider  that  the velocities and the masses of particles that participate in the collisions to be elements of a division ring.
Also, as we shall see, in order to obtain the non-abelian elastic collision map, the  solution of   a system of Sylvester equations is required.
In this Section, we address all these just mentioned issues and we investigate the properties of the obtained maps. 
 First we recall the definition of Yang-Baxter maps and the definition of $3D-$compatible maps
and  we show how these notions are mutually related. 
 Next, we show that the  maps (\ref{CM_EC}) and their non-abelian versions are  $3D-$compatible,  whereas the  maps (\ref{YB_EC}) and their non-abelian versions are Yang-Baxter. 

We would like to mention that the conservation relations (\ref{Re_1})-(\ref{Re_3}), which are equivalent to the conservation of the relativistic momentum and energy, can be considered as the defining relations of the following map
\begin{align*}
R_{ij}:(x^i,x^j)\mapsto (x^i{'},x^j{'})=\left(x^j\frac{m^i_0x^i+m^j_0x^j}{m^j_0x^i+m^i_0x^j},x^i\frac{m^i_0x^i+m^j_0x^j}{m^j_0x^i+m^i_0x^j}\right).
\end{align*}
This map is equivalent with a Yang--Baxter map that was referred to as $H_{III}^A$ in \cite{Papageorgiou:2010}.    While relations (\ref{ext_h3a}) can be considered as the defining relations of a two-component Yang--Baxter map that extends  $H_{III}^A$. This two-component map explicitly reads
\begin{align*}
\hat R_{ij}:(x^i,y^i,p^i;x^j,y^j,p^j)\mapsto & (x^i{'},y^i{'},p^i{'};x^j{'},y^j{'},p^j{'})\\
&=\left(x^j\frac{m^i_0x^i+m^j_0x^j}{m^j_0x^i+m^i_0x^j},y^j+A,-p^i;x^i\frac{m^i_0x^i+m^j_0x^j}{m^j_0x^i+m^i_0x^j
},y^i+A,-p^j\right),
\end{align*}
 where 
   \begin{align*}
   A:=-\left((m^i_0)^2-(m^j_0)^2\right)\frac{x^ix^j(y^i-y^j)}{(m^i_0x^i+m^j_0x^j)(m^j_0x^i+m^i_0x^j)}+x^j\left(\frac{m^j_0p^i}{m^i_0x^i+m^j_0x^j}
  +\frac{m^i_0p^j}{m^j_0x^i+m^i_0x^j}\right).
   \end{align*}

\subsection{ Yang-Baxter  and $3D-$compatible maps  } \label{section2}
%%%%%%%%%%%%

 Let $\mathcal{X}$ be any set. 

 \begin{definition}[$3D-$compatible maps \cite{ABS:YB}]\label{Def1}
 Let $Q: \mathcal{X} \times \mathcal{X}\ni({\bf x},{\bf y})\mapsto ({\bf u}, {\bf v})=(f({\bf x},{\bf y}),g({\bf x},{\bf y})) \in \mathcal{X} \times \mathcal{X},$ be a map and  $Q_{ij}$  $i\neq j\in\{1,2,3\},$ be the maps that act as $Q$ on the $i-$th and $j-$th factor of $\mathcal{X} \times \mathcal{X}\times \mathcal{X}$ and as identity to the remaining factor.  In detail we have
\begin{align*}
Q_{12}:({\bf x},{\bf y},{\bf z})\mapsto( {\bf x}_2, {\bf y}_1,{\bf z})=(f({\bf x},{\bf y}),g({\bf x},{\bf y}),{\bf z}),\\
Q_{13}:({\bf x},{\bf y},{\bf z})\mapsto( {\bf x}_3,{\bf y},{\bf z}_1)=(f({\bf x},{\bf z}),{\bf y},g({\bf x},{\bf z})),\\
Q_{23}:({\bf x},{\bf y},{\bf z})\mapsto({\bf x}, {\bf y}_3, {\bf z}_2)=({\bf x},f({\bf y},{\bf z}),g({\bf y},{\bf z})).
\end{align*}

The map $Q: \mathcal{X} \times \mathcal{X}\rightarrow \mathcal{X} \times \mathcal{X}$ will be called {\em 3D-compatible} or {\em 3D-consistent map}  if it holds
${{\bf x}_{23}}={ {\bf x}_{32}},$  ${ {\bf y}_{13}}={{\bf y}_{31}},$ ${{\bf z}_{12}}={{\bf z}_{21}},$ that is
\begin{align}\label{3d:comp:def1}
f( {\bf x}_3,{\bf y}_3)=f({\bf x}_2,{\bf z}_2),&&g( {\bf x}_3, {\bf y}_3)=f( {\bf y}_1, {\bf z}_1),&&g( {\bf x}_2, {\bf z}_2)=g( {\bf y}_1,{\bf z}_1).
\end{align}
\end{definition}

\begin{definition}[Yang-Baxter maps \cite{Sklyanin:1988,Drinfeld:1992}]
A map $R: \mathcal{X} \times \mathcal{X}\ni({\bf x},{\bf y})\mapsto ({\bf u}, {\bf v})=(s({\bf x},{\bf y}),t({\bf x},{\bf y}))\in \mathcal{X}\times \mathcal{X},$ will be called a {\em Yang-Baxter map} if it satisfies %the {\em Yang-Baxter relation}
\begin{align} \label{YANG_BAXTER}
R_{12}\circ R_{13}\circ R_{23}= R_{23}\circ R_{13}\circ R_{12},
\end{align}
where $R_{ij}$ $i\neq j\in\{1,2,3\},$ denotes the maps that act as  $R$ on the $i-$th and the $j-$th factor of $\mathcal{X}\times \mathcal{X}\times \mathcal{X},$ and as identity to the remaining factor.
\end{definition}

The first instances of Yang-Baxter maps  appeared in \cite{Sklyanin:1988,Drinfeld:1992}.  The term {\em Yang-Baxter maps} was introduced in \cite{Bukhshtaber:1998,Veselov:20031} to refer to maps that serve as set-theoretical solutions of the Yang--Baxter equation. The Yang--Baxter equation
originally appeared  in quantum physics \cite{McGuire_1964,Yang:1967} and in statistical mechanics \cite{Baxter:1972,Baxter:1982} and it was 
further noticed 
that  it is  equivalent to the braid group relations \cite{Artin_1925,Artin_1947}. Within this identification, the  Yang-Baxter theory emerged in various areas of mathematics such as
quantum groups (Hopf algebras,
von Neumann algebras), knot theory and, most importantly
  from the point of view of this article,  
the theory of discrete integrable systems \cite{Franks-book}. 

\begin{definition}[Birational maps]
An invertible   map $R: \mathcal{X} \times \mathcal{X}\ni({\bf x},{\bf y})\mapsto ({\bf u},{\bf v}) \in \mathcal{X} \times \mathcal{X}$   will be called {\em birational}, if both the map $R$ and its inverse $R^{-1}: \mathcal{X} \times \mathcal{X}\ni({\bf u},{\bf v})\mapsto ({\bf x},{\bf y}) \in \mathcal{X} \times \mathcal{X},$ are rational maps.
\end{definition}

\begin{definition}[Quadrirational maps and their companion maps \cite{Etingof_1999,ABS:YB}]
A map $R: \mathcal{X} \times \mathcal{X}\ni({\bf x},{\bf y})\mapsto ({\bf u},{\bf v}) \in \mathcal{X} \times \mathcal{X}$   will be called {\em quadrirational}, if both the map $R$ and the so-called {\em companion map} $R^c: \mathcal{X} \times \mathcal{X}\ni({\bf x},{\bf v})\mapsto ({\bf u},{\bf y}) \in \mathcal{X} \times \mathcal{X},$ are birational maps.
\end{definition}

%%%%%%%%%%%%
\begin{prop}\label{prop_ent1} \cite{ABS:YB}
Map $Q:({\bf x},{\bf y})\mapsto \left(f({\bf x},{\bf y}),g({\bf x},{\bf y})\right)$ is a $3D-$compatible map, iff 
its companion map is Yang Baxter map.
\end{prop}

%%%%%%%%%%%%%%%%%%%%

%\input{fig1.tex}

%%%%%%%%%%%%%%%%%
%%%%%%%%%%%%%%%%%
%%%%%%%%%%%%%%%%%

The  definitions above can be easily extended to $N$-dimensions with $N>3.$

 \begin{definition}[Multidimensionally compatible maps  \cite{ABS:YB}]\label{multiDef}
 Let $Q: \mathcal{X} \times \mathcal{X}\ni({\bf x},{\bf y})\mapsto ({\bf u}, {\bf v})=(f({\bf x},{\bf y}),g({\bf x},{\bf y})) \in \mathcal{X} \times \mathcal{X},$ be a map and  $Q_{ij}$  $i\neq j\in\{1,\ldots,N\},$ be the maps that act as $Q$ on the $i-$th and $j-$th factor of $\mathcal{X}^n $ and as identity to the remaining factor.  In detail we have
\begin{align*}
Q_{ij}:&
({\bf x^1},\ldots,{\bf x}^i,\ldots,{\bf x}^j,\ldots,{\bf x}^n)
\mapsto 
({\bf x}^1,\ldots,{\bf x}^i_j ,\ldots,{\bf x}^j_i, \ldots, {\bf x}^n)&\\
{}&=({\bf x}^1,\ldots,f({\bf x}^i,{\bf x}^j) ,\ldots,g({\bf x}^i,{\bf x}^j), \ldots, {\bf x}^n), &  i\neq j \in \{1,\ldots, N\}
\end{align*}
The map $Q: \mathcal{X} \times \mathcal{X}\rightarrow \mathcal{X} \times \mathcal{X}$ will be called {\em multidimensionally compatible}   if it holds
\begin{align*}
  {\bf x}^i_{jk}=&{\bf x}^i_{kj},&   i\neq j\neq k\neq i\in \{1,\ldots, N\}.
\end{align*}
\end{definition}
\begin{remark}
The companion map $Q^c$ of the multidimensionally compatible map $Q$, satisfies the following Yang-Baxter equations
\begin{align*}
 Q^c_{ij}\circ Q^c_{ik}\circ Q^c_{jk}=& Q^c_{jk}\circ Q^c_{ik}\circ Q^c_{ij},&  i\neq j\neq k\neq i\in \{1,\ldots, N\}.
\end{align*}
\end{remark}

\subsection{The non-abelian  companion elastic collision maps are multidimensionally compatible} \label{topsub}

 Here we prove that the non-abelian extension of (\ref{CM_EC}) that is system (\ref{TheSystemI}),   is multidimensionally compatible. 
 First, we need to extend (\ref{TheSystemI}) to multi-dimensions, that can be done by the following identification
  \begin{align*}
   \left(v^{i},\m^{i},v^{i}{'},\m^{i}{'};v^{j},\m^{j},v^{j}{'},\m^{j}{'}\right) \equiv &  \left(v^{i},\m^{i},v^{i}_{j},\m^{i}_{j};v^{j}_i,\m^{j}_i, v^{j},\m^{j}\right), & i\neq j \in \{1,\ldots, N\}.
   \end{align*}
   In terms of this identification (\ref{TheSystemI}) reads
   \begin{subequations}\label{def_eq}
%\begin{empheq}[left=\empheqlbrace]{align}
\begin{gather} \label{eq:2:1}
 v^{i}_j-v^{j}=v^{j}_i-v^{i},\\ \label{eq:2:2}
\m^{i}_j+\m^{j}=\m^{j}_i+\m^{i},\\ \label{eq:2:3}
\m^{i}_j\m^{j}=\m^{j}_i\m^{i},\\ \label{eq:2:4}
\m^{i}_jv^{j}+v^{i}_j\m^{j}=\m^{j}_iv^{i}+v^{j}_i\m^{i},
\end{gather}
%&\end{empheq}
\end{subequations}
where $i\neq j \in \{1,\ldots, N\}.$ This is a linear set of equations with respect to $v^{i}_j, v^{j}_i, \m^{i}_j, \m^{j}_i.$ Solving this set of equations for
$v^{i}_j, v^{j}_i, \m^{i}_j, \m^{j}_i,$ we obtain  the defining formulas for  the  maps $Q_{ij}$.
Indeed, the non-abelian maps $Q_{ij}$ explicitly read
\begin{align} \label{CM_MAPS_0}
Q_{ij}:(v^{i},\m^{i}; v^{j},\m^{j})\mapsto&  \left(v^{i}_{j},\m^{i}_{j}; v^{j}_i,\m^{j}_i\right), & i\neq j \in \{1,\ldots, N\},
\end{align}
where
%\begin{subequations*}
%\begin{empheq}[left=\empheqlbrace]{align}
\begin{align}\label{CM_MAPS_12}
\m^{i}_j=&K^{i,j}\m^{i}\left(K^{i,j}\right)^{-1},\\ \label{CM_MAPS_120}
 v^{i}_{j}=&\left(\m^{i}_j v^{j}-\m^{j}_i v^{i}-\left(v^{i}-v^{j}\right)\m^{i}\right)\left(K^{i,j}\right)^{-1},
 \end{align}
%\end{empheq}
%\end{subequations*}
and the expressions $K^{i,j}$ are defined by the formulas  \[K^{i,j}:=\m^{i}-\m^{j}.\] 
We will need the following lemma.
\begin{lemma} \label{1st_lemmam}
Consider the maps (\ref{CM_MAPS_0}), then the following relations hold
\begin{enumerate}
\item $K^{i,j}_k+K^{j,k}_i+K^{k,i}_j=0$ (additive closure relation);%$\m^{i}_j-\m^{j}_i+\m^{j}_k-\m^{k]}_j+\m^{k}_i-\m^{i}_k=0$;
\item $\m^{j}_k\left(\m^{i}_k\right)^{-1}\m^{k}_i=\m^{k}_j\left(\m^{i}_j\right)^{-1}\m^{j}_i$ (multiplicative closure relation).
\item The expressions
\begin{align} \label{def_g}
{}^i\Gamma^{j,k}:=K^{i,j}_k K^{i,k}%\left(\m^{i}_k-\m^{j}_k\right)K^{i,j},
\end{align}
are symetric with respect to the interchange $j$ to $k,$ and they satisfy
\begin{align}\label{lemma_31}
{}^i\Gamma^{j,k}\m^{i}\left({}^i\Gamma^{j,k}\right)^{-1} {}^i\Psi^{j,k}+{}^i\Omega^{j,k}=0
\end{align}
where ${}^i\Psi^{j,k}$ and ${}^i\Omega^{j,k}$ are defined as follows
\begin{align}\label{psi}
&{}^i\Psi^{j,k}:=\left(1-\m^{j}_k\left(\m^{i}_k\right)^{-1}\right)\left(v^{k}-
v^{i}\right)-\left(1-\m^{k}_j\left(\m^{i}_j\right)^{-1}\right)\left(v^{j}-v^{i}\right),\\ \label{ksi}
&{}^i\Omega^{j,k}:=\left(\m^{i}_j-\m^{k}_j\right)\left(v^{j}-v^{i}\right)-\left(\m^{i}_k-\m^{j}_k\right)\left(v^{k}-v^{i}\right),
\end{align}
which are clearly antisymmetric with respect to the interchange $j$ to $k.$
\end{enumerate}
\end{lemma}
\begin{proof}
The proof of the lemma is given in Appendix \ref{app0}.
\end{proof}

Now we are ready to formulate the main  Theorem of this Section.

\begin{theorem} \label{theo1m}
Maps (\ref{CM_MAPS_0}) are multidimensionally compatible.
\end{theorem}

\begin{proof}
 The proof of this theorem is presented in Appendix \ref{app00}. 
\end{proof}

\subsection{Yang-Baxter maps and the Sylvester equation}
Theorem \ref{theo1m} states  that maps (\ref{CM_MAPS_0}) are multidimensionally-compatible. In the following proposition we present the companion maps of  (\ref{CM_MAPS_0}) which are Yang-Baxter maps due to Proposition \ref{prop_ent1}.
\begin{prop}
The companion maps $Q_{ij}^c$ of (\ref{CM_MAPS_0}) read
\begin{align} \label{comp_map}
Q_{ij}^c:(v^{i},\m^{i}; v^{j}_i,\m^{j}_i)\mapsto \left(v^{i}_{j},\m^{i}_{j}; v^{j},\m^{j}\right),&& i\neq j \in \{1,\ldots, N\},
\end{align}
where
\begin{align}\label{def_comp}
\begin{aligned}
  v^{i}_{j}=v^{j}_{i}+h^{i,j}, && v^{j}=v^{i}+h^{i,j},\\
  \m^{i}_{j}=\m^{j}_{i}+\left(g^{i,j}\right)^{-1}, && \m^{j}=\m^{i}-\left(g^{i,j}\right)^{-1},
  \end{aligned}
\end{align}
and $h^{i,j}, g^{i,j}$ satisfy the following system of Sylvester equations
\begin{gather}\label{sylv}
\begin{gathered}
\m^ig^{i,j}-g^{i,j}\m^j_i=1,\\
\left(\m^j_i+\left(g^{i,j}\right)^{-1}\right)h^{i,j}+h^{i,j}\left(\m^i-\left(g^{i,j}\right)^{-1}\right)=u^j_i\left(g^{i,j}\right)^{-1}-\left(g^{i,j}\right)^{-1}u^i.
\end{gathered}
\end{gather}
\end{prop}
\begin{proof}
In order to find the companion maps $Q_{ij}^c$ we need to solve (\ref{def_eq}) for $v^{i}_{j},\m^{i}_{j}, v^{j}, \m^{j}$ in terms of $v^{i},\m^{i}, v^{j}_i,\m^{j}_i.$ Omitting the identity solution $v^{i}_{j}=v^{j}_{i},\m^{i}_{j}=\m^{j}_{i}, v^{j}=v^{i}, \m^{j}=\m^{i},$ we consider the auxiliary variables $g^{i,j},$ and $h^{i,j},$ defined by
\begin{gather}\label{auxi}
  \m^{i}_{j}=\m^{j}_{i}+(g^{i,j})^{-1},\quad v^{i}_{j}=v^{j}_{i}+h^{i,j}.
\end{gather}
Substituting the expressions above into (\ref{eq:2:1}) and (\ref{eq:2:2}), we respectively obtain
\begin{gather}\label{auxi2}
  \m^{j}=\m^{i}-(g^{i,j})^{-1},\quad v^{j}=v^{i}+h^{i,j}.
\end{gather}
Substituting  (\ref{auxi}) and (\ref{auxi2}) into (\ref{eq:2:2}) and (\ref{eq:2:3}), they become exactly the system of Sylvester equations  (\ref{sylv}) and that completes the proof.
\end{proof}
We remark that in a similar manner we can find the inverse of the maps $Q_{ij}^c$, hence the original maps $Q_{ij}$ are quadrirational. Further recent developments on non-Abelian Yang-Baxter maps can be found in \cite{Noumi:2020,Kassotakis:2:2021,Kassotakis:2022b,Rizos:2024}.

%%%%%%%%%%%%%%%%%
%%%%%%%%%%%%%%%%%
\section{Non-Abelian elastic collision maps as difference systems}\label{re0}
%%%%%%%%%%%%%%%%%
%%%%%%%%%%%%%%%%%

\subsection{ Reinterpretation of a multidimensionaly compatible map as a system of difference equations on the ${\mathbb Z}^N$ graph} \label{re}
There is a natural association of a map with a difference system  defined on the edges of an elementary quadrilateral of the $\mathbb{Z}^2$ graph. Indeed, a map $R: \mathcal{X} \times \mathcal{X}\ni (\bf x,\bf y) \mapsto  (\bf x{'},\bf y{'}) \in  \mathcal{X} \times \mathcal{X},$ can be considered as a difference system defined on the edges of an elementary quadrilateral of the $\mathbb{Z}^2$ graph  by  making the following identifications
\begin{align} \label{notation1}
\begin{aligned}
{\bf x}\equiv  {\bf x}_{m+1/2,n},&&{\bf y}\equiv  {\bf y}_{m,n+1/2}\\
{\bf x{'}}\equiv  {\bf x}_{m+1/2,n+1},&&{\bf y{'}}\equiv  {\bf y}_{m+1,n+1/2},
\end{aligned} && m,n\in \mathbb{Z},
\end{align}
 so that the ``primes" have been interpreted as increments on the independent variables.
Moreover, we can adopt the compendious notation (see Figure \ref{fig1})
\begin{align*}
\begin{aligned}
{\bf x}:={\bf x}_{m+1/2,n},&& {\bf y}:={\bf y}_{m,n+1/2}, && {\bf x}_1:={\bf x}_{m+3/2,n}, && etc.\\
 {\bf x}_2:={\bf x}_{m+1/2,n+1}\equiv {\bf u},  &&{\bf y}_1:={\bf y}_{m+1,n+1/2}\equiv {\bf v}, && {\bf y}_2:={\bf y}_{m,n+3/2},&& etc.
\end{aligned} && m,n\in \mathbb{Z},
\end{align*}
where with subscripts we denoted discrete shifts on the associated $\mathbb{Z}^2$ graph. We extend our notation to the $\mathbb{Z}^N$ graph as follows
\begin{align} \label{notation1.1}
\begin{aligned}
{\bf x}^{i}:={\bf x},&& {\bf x}^{j}:={\bf y},&&{\bf x}^{i}_j:={\bf x}_j,&& {\bf x}^{j}_i:={\bf y}_i,
\end{aligned} && i\neq j \in \{1,\ldots, N\}.
\end{align}
 Note that in this notation the superscripts  represent the associated edges of the $\mathbb{Z}^N$ graph where the variables are assigned, e.g.  see Figure \ref{fig1}.
When ${\bf x}$ represents an $M\in \mathbb{N}$ component vector, so ${\bf x}^{i}$ stands for the vector ${\bf x}$ assigned on the $i-$th edge of the $\mathbb{Z}^N$ graph,  we denote its components as follows
\begin{align*}
  {\bf x}^{i}=\left(x^{1,i},\ldots, x^{M,i}\right),&&i \in \{1,\ldots, N\},&& M\in \mathbb{N}.
\end{align*}
In this concise notation introduced above, the difference system $ {\bf x}^{i}_j=F( {\bf x}^{i}, {\bf x}^{j}),$  is associated with the map $({\bf x}^{i},{\bf x}^{j})\mapsto \left({\bf x}^{i}_j,{\bf x}^{j}_i\right)=\left(F({\bf x}^{i},{\bf x}^{j}),F({\bf x}^{j},{\bf x}^{i})\right),$ where $F$ a rational function. This difference system will be called multidimensional compatible  {\em iff}
\begin{align} \label{mult-com}
 {\bf x}^{i}_{jk}=&{\bf x}^{i}_{kj}, & i\neq j\neq k\neq i\in \{1,\ldots, N\}.
\end{align}

\begin{figure}[h]
\begin{center}
\begin{minipage}[htb]{0.4\textwidth} \tikzcdset{every label/.append style = {font = \small}}
\adjustbox{scale=0.85,center}{
\begin{tikzcd}[row sep=1.5in, column sep = 1.5in,every arrow/.append style={dash}]
  \arrow[d,"\bf{y^i}\equiv \bf{y}_{m,n+1/2}" {description},crossing over,sloped,anchor=north]  \bff{\phi}_{m,n+1}  \arrow[r, "\bff{x{'}}\equiv\bf{x}_{m+1/2,n+1}" {description},crossing over]&   \bff{\phi}_{m+1,n+1}\arrow{d} \\
   \bff{\phi}_{m,n} \arrow[r, "\bf{x}\equiv\bf{x}_{m+1/2,n}" {description},crossing over]&  \bff{\phi}_{m+1,n} \arrow[u," \bf{y{'}}\equiv\bf{y}_{m+1,n+1/2}" {description},crossing over,sloped]
\end{tikzcd}}
\captionsetup{font=footnotesize}%,calcwidth=0.8\columnwidth
\captionof*{figure}{(a) Descriptive notation}
\end{minipage}\hspace{0.9cm}
\begin{minipage}[htb]{0.4\textwidth} \tikzcdset{every label/.append style = {font = \large}}
\adjustbox{scale=0.85,center}{
\begin{tikzcd}[row sep=1.5in, column sep = 1.5in,every arrow/.append style={dash}]
  {\bff \phi}_2 \arrow[r,"{\bf x}_2" {description},crossing over]\arrow[d,"{\bf y}" {description},crossing over]&  {\bff \phi}_{12}\arrow{d} \\
  {\bff \phi} \arrow[r,"{\bf x}" {description},crossing over]& {\bff \phi}_1\arrow[u,"{\bf y}_1" {description},crossing over]
\end{tikzcd}}
\captionsetup{font=footnotesize}%,calcwidth=0.8\columnwidth
\captionof*{figure}{(b) Compendious notation}
\end{minipage}\
\caption{Variables assigned on vertices and edges of  an elementary cell of the $\mathbb{Z}^2$ graph} \label{fig1}
\end{center}
\end{figure}
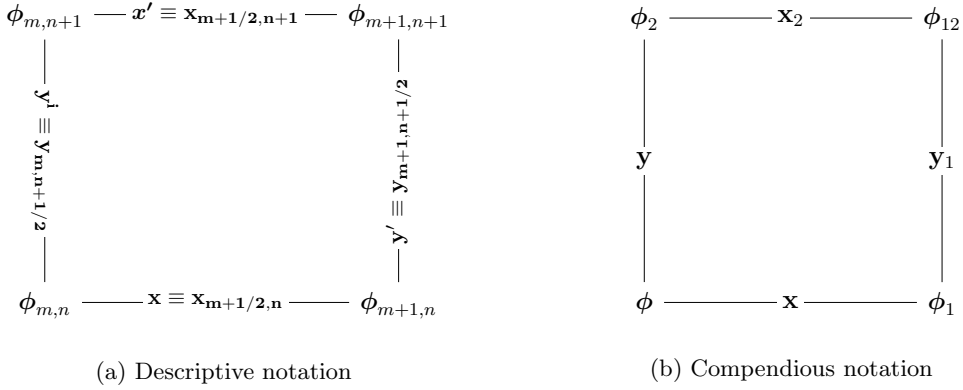
For the rest of this article and unless else specified,  we have
\begin{align*}
 {\bf x}^i:=&{\bf x}=\left(\mu^i,v^i\right), & {\bf x}^j:=&{\bf y}=\left(\mu^j,v^j\right), & i\neq j \in \{1,\ldots, N\}.
\end{align*}

 With the reinterpretation considered above,
the abelian version of elastic collision map appeared in the discrete integrable systems theory in \cite{Kouloukas-2017,Dimakis2019,KOULOUKAS:2023}, see also \cite{KaNie:2018} for their connection with discrete Hirota's Korteweg de Vries equation.

\subsection{Non-Abelian  companion elastic collision maps and associated difference systems}
%%%%%%%%%%%%%%%%%%%%%%%%%%%%%%%%%%%%%%%%%%%%%%%%%%%%%%%%%%%%%%%%%%%%%%%%%%%%%%%%%%
\label{nelco}

With the reinterpretation  we presented in the previous Section, we can treat the defining relations of the  maps  (\ref{CM_MAPS_0}) that is (\ref{CM_MAPS_12}) and (\ref{CM_MAPS_120}),
as  difference systems that relate the two-component variables defined on the edges of an elementary quad of the ${\mathbb Z}^N$ graph.

\begin{theorem} \label{theo1}
Consider the following  Lax system
\begin{align} \label{cc:lp}
\Psi_i= L(v^{i},\m^{i};\lambda) \Psi, && \Psi_j= L(v^{j},\m^{j};\lambda) \Psi,
\end{align}
where the Lax matrix $L$ is given by
\begin{align*}
L(v^{i},\m^{i};\lambda):=\begin{pmatrix}
\m^{i}+\lambda&\lambda v^{i}\\
0&\m^{i}-\lambda
\end{pmatrix}.
\end{align*}
Let all entries of the Lax matrix be assumed to belong to an associative algebra $\mathcal{A}$, and
  $\lambda$, which is referred to as the spectral parameter,  assumed to be an element of the center of the algebra. The following holds.
\begin{enumerate}
\item  The compatibility conditions of the  Lax system (\ref{cc:lp}) i.e.
\begin{equation} \label{linearp}
L(v^{i}_j,\m^{i}_j;\lambda)L(v^{j},\m^{j};\lambda)=L(v^{j}_i,\m^{j}_i;\lambda)L(v^{i},\m^{i};\lambda)
\end{equation}
that are supposed to be valid for every value of $\lambda$ are equivalent to the following system of difference equations defined on the edges of the $\mathbb{Z}^N$ graph
\begin{subequations}\label{def_eqq}
%\begin{empheq}[left=\empheqlbrace]{align}
\begin{gather} \label{eqq:2:1}
 v^{i}_j-v^{j}=v^{j}_i-v^{i},\\ \label{eqq:2:2}
\m^{i}_j+\m^{j}=\m^{j}_i+\m^{i},\\ \label{eqq:2:3}
\m^{i}_j\m^{j}=\m^{j}_i\m^{i},\\ \label{eqq:2:4}
\m^{i}_jv^{j}+v^{i}_j\m^{j}=\m^{j}_iv^{i}+v^{j}_i\m^{i},
\end{gather}
%&\end{empheq}
\end{subequations}
where $i\neq j \in \{1,\ldots, N\};$ 
\item they corresponds to the following  systems of difference equations  defined on the vertices ({\em vertex systems}) of the $\mathbb{Z}^N$ graph 
\begin{subequations}\label{top}
\begin{gather}
\psi_i-\psi=\phi_i\phi^{-1},\\
%\psi_j-\psi=\phi_j\phi^{-1},\\
\sigma_i-\sigma=\phi_i^{-1}\phi,\\
\chi_i+\chi=\phi_i(\omega_i-\omega)\phi^{-1}.
%\chi_j+\chi=\phi_j(\omega_j-\omega)\phi^{-1}.
\end{gather}
\end{subequations}
\item  The  relations (\ref{top}), serve as hetero-B\"acklund transformations between the following three vertex systems
\begin{subequations}\label{vs2}
%\begin{empheq}[left=\empheqlbrace]{align}\label{vs22}
\begin{gather}\label{vs22}
\phi_{ij}\left(\phi_j^{-1}-\phi_i^{-1}\right)=\left(\phi_i-\phi_j\right)\phi^{-1},\\ \label{vs21}
\phi_{ij}\phi_j^{-1}(\chi_j+\chi)+(\chi_{ij}+\chi_j)\phi_j\phi^{-1}=\phi_{ij}\phi_1^{-1}(\chi_i+\chi)+(\chi_{ij}+\chi_i)\phi_i\phi^{-1},
\end{gather}
%\end{empheq}
\end{subequations}
and
\begin{subequations}\label{vs3}
%\begin{empheq}[left=\empheqlbrace]{align}
\begin{gather}\label{vs32}
\phi_{ij}\left(\phi_j^{-1}-\phi_i^{-1}\right)=\left(\phi_i-\phi_j\right)\phi^{-1},\\ \label{vs31}
\phi_{ij}(\omega_{ij}-\omega_j)\phi_j^{-1}+\phi_{i}(\omega_{i}-\omega)\phi^{-1}=\phi_{ij}(\omega_{ij}-\omega_i)\phi_i^{-1}+\phi_{j}(\omega_{j}-\omega)\phi^{-1},
\end{gather}
%\end{empheq}
\end{subequations}
and
\begin{subequations}\label{vs1}
%\begin{empheq}[left=\empheqlbrace]{align}
\begin{gather}\label{vs12}
(\psi_{ij}-\psi_j)(\psi_j-\psi)=(\psi_{ij}-\psi_i)(\psi_i-\psi),\\ \label{vs11}
(\psi_{ij}-\psi_j)(\chi_j+\chi)+(\chi_{ij}+\chi_j)(\psi_j-\psi)=(\psi_{ij}-\psi_i)(\chi_i+\chi)+(\chi_{ij}+\chi_i)(\psi_i-\psi).
\end{gather}
%\end{empheq}
\end{subequations}
\item 
The difference systems (\ref{def_eqq}) are linear  with respect to the variables $v^{i}_j, v^{j}_i, \m^{i}_j, \m^{j}_i$  and  provided that  the expressions $\m^{i}-\m^{j}$ are  invertible,
they can be solved with respect to these variables to obtain
\begin{subequations} \label{subeq_0}
%\begin{empheq}[left=\empheqlbrace]{align}\label{CM_MAPS_11}
\begin{gather}\label{CM_MAPS_11_n}
\m^{i}_j =(\m^{i}-\m^{j})\m^{i} (\m^{i}-\m^{j})^{-1},\\ \label{CM_MAPS_12_n}
v^{i}_{j}=(\m^{i}_j v^{j}-\m^{j}_i v^{i}-\left(v^{i}-v^{j}\right)\m^{i})(\m^{i}-\m^{j})^{-1}.
\end{gather}%\end{empheq}
\end{subequations}
The system of equations in edge variables (\ref{def_eqq}) %(\ref{subeq_0}) 
are multidimensionally  compatible.
\item The  systems of equations in vertex variables  (\ref{top}) are multidimensional compatible.
\end{enumerate}

\end{theorem}

\begin{proof}
\begin{enumerate}
\item The proof follows by direct computation. Indeed, the compatibility conditions (\ref{linearp}) of the  Lax system (\ref{cc:lp}) explicitly read
    \begin{subequations} \label{cc_c}
    \begin{align} \label{cc_c1}
    (\m^i_j+\lambda)(\m^j+\lambda)=&(\m^j_i+\lambda)(\m^i+\lambda),\\ \label{cc_c2}
    (\m^i_j-\lambda)(\m^j-\lambda)=&(\m^j_i-\lambda)(\m^i-\lambda),\\ \label{cc_c3}
    (\m^i_j+\lambda)v^j+v^i_j(\m^j-\lambda)=&(\m^j_i+\lambda)v^i+v^j_i(\m^i-\lambda).
    \end{align}
    \end{subequations}
Demanding that  (\ref{cc_c})  hold for every $\lambda,$  (\ref{cc_c1}) is equivalent to (\ref{cc_c2}) and together with (\ref{cc_c3}) we obtain exactly
the system of differences equations (\ref{def_eqq}).
\item 
The procedure to obtain a system of vertex equations associated with its edge system counterpart,   
is nowadays referred to as {\em potentialization}
 \cite{Kassotakis_2011,Kassotakis_2012,Kouloukas:2012,Doliwa_2013,Fordy_2017,Kassotakis_2021,Kass2}.
Namely, equations (\ref{eqq:2:1}),  (\ref{eqq:2:2}) and (\ref{eqq:2:3})  guarantee  the existence of the potential functions   $\chi,$  $\psi,$  $\phi,$ respectively
\begin{subequations}\label{xvw}
\begin{align}
\label{x}
  &v^{i}=\chi_i+\chi,& v^{j}=\chi_j+\chi,& \\ 
  \label{v}
  &\m^{i}=\psi_i-\psi,& \m^{j}=\psi_j-\psi,& \\ 
  \label{phi}
&\m^{i}=\phi_i\phi^{-1},& \m^{j}=\phi_j\phi^{-1},& 
  \end{align}
\end{subequations}
Then equation (\ref{eqq:2:4}) can be written in the conservation form
\begin{gather}\label{omega}
\left(\phi_i^{-1}v^i\phi\right)_j-\phi_i^{-1}v^i\phi=\left(\phi_j^{-1}v^j\phi\right)_i-\phi_j^{-1}v^j\phi.
\end{gather}
Moreover, (\ref{omega})  guarantees the existence of the potential function
$\omega$ which is defined by
\begin{gather}
\label{w}
  \phi_i^{-1}v^i\phi=\omega_i-\omega, \qquad \phi_j^{-1}v^j\phi=\omega_j-\omega.
\end{gather}
Furthermore, taking $$\rho:=\phi^{-1},$$  due to (\ref{phi}) equation (\ref{eqq:2:2}) becomes the conservation relation
\begin{align}\label{sigma_0}
\left(\rho_i\rho^{-1}\right)_j-\rho_i\rho^{-1}=\left(\rho_j\rho^{-1}\right)_i-\rho_j\rho^{-1},
\end{align}
that guarantees the existence of the following potential $\sigma$
\begin{align}\label{sigma}
&\sigma^i-\sigma=\rho_i\rho^{-1}, %=\phi_i^{-1}\phi,
& \sigma^j-\sigma=\rho_j\rho^{-1}.& %=\phi_j^{-1}\phi
\end{align}
From   (\ref{xvw}) and (\ref{omega}), we eliminate $v^{i}, v^{j}$, $\m^{i}, \m^{j}$ and together with (\ref{sigma}) we obtain  (\ref{top}).
\item
From the difference systems in vertex variables (\ref{top}), we use the compatibility conditions $\phi_{ij}=\phi_{ji},$ $\chi_{ij}=\chi_{ji},$ $\psi_{ij}=\psi_{ji},$ $\omega_{ij}=\omega_{ji},$ to either eliminate $\psi$ and $\omega$, obtaining (\ref{vs2}); eliminate $\psi$ and $\chi$, obtaining (\ref{vs3}); or eliminate $\phi$ and $\omega$, obtaining (\ref{vs1}).

\item The proof of this fact is  essentially the same as the proof of Theorem \ref{theo1m}.

\item
In order to prove that the vertex systems  (\ref{top}) are multidimensional compatible we have to show that 
$\forall i\neq j \neq k \in\{1,\ldots, N\},$ 
$\phi_{ijk}=\phi_{ikj},$ $\chi_{ijk}=\chi_{ikj},$ $\psi_{ijk}=\psi_{ikj}$ and $\omega_{ijk}=\omega_{ikj}$. Let us first prove that  $\phi_{ijk}=\phi_{ikj}.$ From item $(4)$ of this theorem we have  that $\m^i_{jk}=\m^i_{kj}.$ Substituting the definition of the potential $\phi$ (\ref{phi}) to the previous relations we obtain $\phi_{ijk}\phi_{jk}^{-1}=\phi_{ikj}\phi_{kj}^{-1}.$ So in order $\phi_{ijk}=\phi_{ikj}$ to hold, it should be that $\phi_{jk}=\phi_{kj},$ but this trivially holds since the potential $\phi$ exists. Similarly we prove that the remaining potentials are multidimensional consistent.

\end{enumerate}
\end{proof}

Note that,   equations (\ref{eqq:2:2}),(\ref{eqq:2:3}) also appear
in the context of discrete differential geometry and  serve as the defining relations of the so-called {\em skew parallelogram nets} \cite{Schief_2007,Bobenko:2008},  c.f. \cite{Hoffmann:2024}. In that respect (\ref{def_eqq}) could be realized as  deformations of skew parallelogram nets since for $v^i=v^i_j=v^j=v^j_i=c,$ where the constant $c$ is an element of the field that the algebra is defined over, we recover (\ref{eqq:2:2}),(\ref{eqq:2:3}).
In \cite{Doliwa_2013,Doliwa_2014,Noumi:2020,Kassotakis:1:2021,Kassotakis:2022b}, reductions of (\ref{eqq:2:2}),(\ref{eqq:2:3}) on certain subspaces of  a $\mathbb{Z}_n-$graded algebra  over a non-commutative ring were considered together with the associated non-abelian Yang-Baxter maps and the corresponding difference systems in edge and in vertex variables.

\subsection{Abelian elastic collisions of point-mass particles as difference  equations on the $\mathbb{Z}^N$ graph}
%%%%%%%%%%%%%%%%%%%%%%%%%%%%%%%%%%%%%%%%%%%%%%%%%%%%%%%%%%%%%%%%%%%%%%%%%%%%
Under the identifications (\ref{notation1}), (\ref{notation1.1}), the defining relations of the abelian map $Q_{ij}$ (\ref{CM_EC}), define the following difference systems in edge variables
\begin{align} \label{diff_sys_1}
 v^{i}_j-v^{j}=&v^{j}_i-v^{i},& m^{i}\left(v^{i}_j-v^{i}\right)=&m^{j}\left(v^{j}_i-v^{j}\right),& i\neq j\in\{1,\ldots, N\}.
 \end{align}
 The first equation of (\ref{diff_sys_1}), guarantees the existence  of a potential function $\chi,$ such that
 \begin{align*}
 v^i=&\chi_i+\chi,& v^j=&\chi_j+\chi.
 \end{align*}
 In terms of the potential function $\chi$ the second  equation of (\ref{diff_sys_1}), reads
 \begin{align}\label{Vert_sys_1}
\left(m^{i}-m^{j}\right)(\chi_{ij}-\chi)-\left(m^{i}+m^{j}\right)(\chi_i-\chi_j)=0.
\end{align}
On the other hand, The second equation of (\ref{diff_sys_1}), guarantees the existence  of a potential function $\psi,$ such that
 \begin{align*}
 v^i=&\frac{\psi_i-\psi}{m^i},& v^j=&\frac{\psi_j-\psi}{m^j}.
 \end{align*}
 Then, in terms of the potential function $\psi$ the first  equation of (\ref{diff_sys_1}), reads
 \begin{align}\label{Vert_sys_22}
\left(\m^{i}-\m^{j}\right)(\psi_{ij}-\psi)+\left(\m^{i}+\m^{j}\right)(\psi_i-\psi_j)=0,
\end{align}
where $\m^i:=1/m^i$ and $\m^j:=1/m^j.$
The linear difference systems in vertex variables (\ref{Vert_sys_1}) and (\ref{Vert_sys_22}), are related by the substitution
\begin{align*}
\chi_i+\chi=&\frac{\psi_i-\psi}{m^i},& \chi_j+\chi=&\frac{\psi_j-\psi}{m^j}.
\end{align*}

\section{The closure relations and systems of three-dimensional vertex equations} \label{sec4n}

In Lemma \ref{1st_lemmam} we have provided two algebraic relations which are satisfied by the  multidimensional compatible maps (\ref{CM_MAPS_0}) on any cubic-cell of the $Z^N-$graph. Clearly, these algebraic relations also hold for  (\ref{CM_MAPS_11_n}),(\ref{CM_MAPS_12_n}) that serve as the difference systems in edge variables associated with the multidimensional compatible maps (\ref{CM_MAPS_0}). These relations serve as closure relations (c.f. \cite{Lobb:2009}) since they hold on any cubic-cell of the $N-$cube lattice where  the non-abelian systems (\ref{CM_MAPS_11_n}),(\ref{CM_MAPS_12_n}) are defined. 

\subsection{Three-dimensional systems of vertex equations}
In the following Proposition we provide all closure relations associated with the  difference systems (\ref{CM_MAPS_11_n}),  (\ref{CM_MAPS_12_n}), as well as the corresponding systems of three-dimensional vertex equations. As a result we obtain coupled systems of three-dimensional vertex equations.

\begin{prop}
$A.$ On any cubic-cell of the $N-$cube lattice, the solutions of the difference systems (\ref{subeq_0}) %(\ref{CM_MAPS_12}),(\ref{CM_MAPS_11})
 satisfy the following closure relations
\begin{enumerate}[label=(\roman*)]
\item $K^{i,j}_k+K^{j,k}_i+K^{k,i}_j=0,$ where $K^{i,j}:=\m^i-\m^j$;
\item $S^{i,j}_k S^{j,k}_i S^{k,i}_j=1,$ where $S^{i,j}:=\m^i\left(\m^j\right)^{-1}$;
\item $T^{i,j}_k+T^{j,k}_i+T^{k,i}_j=0,$ where $T^{i,j}:=u^i-u^j$;
\item $U^{i,j}_k+U^{j,k}_i+U^{k,i}_j=0,$ where $U^{i,j}:=\phi_i^{-1}u^i\phi-\phi_j^{-1}u^j\phi,$ and $\phi$ the potential function defined in (\ref{phi}).
    \end{enumerate}
    \noindent B. The solutions of the difference systems in vertex variables (\ref{vs2}),(\ref{vs3}),(\ref{vs1}) satisfy the following systems of three-dimensional vertex equations
    \begin{flalign*}
    \begin{aligned}
     \phi_{ij}\phi_{ik}^{-1}\phi_{jk}=\phi_{jk}\phi_{ik}^{-1}\phi_{ij},\\
     \left(\mbox{in terms of $\psi$ reads}:  
   (\psi_{ik}-\psi_{k})(\psi_{jk}-\psi_{k})^{-1}(\psi_{ij}-\psi_{i})(\psi_{ik}-\psi_{i})^{-1}(\psi_{jk}-\psi_{j})(\psi_{ij}-\psi_{j})^{-1}=1\right),\\
    \left(\phi_{ik}^{-1}\left(\chi_{ik}+\chi_k\right)-\phi_{jk}^{-1}\left(\chi_{jk}+\chi_k\right)\right)\phi_k+
    \left(\phi_{ij}^{-1}\left(\chi_{ij}+\chi_i\right)-\phi_{ik}^{-1}\left(\chi_{ik}+\chi_i\right)\right)\phi_i\\
    +\left(\phi_{jk}^{-1}\left(\chi_{jk}+\chi_j\right)-\phi_{ij}^{-1}\left(\chi_{ij}+\chi_j\right)\right)\phi_j=0,
    \end{aligned}&& (\mathcal{X}_2)
    \end{flalign*}
\begin{flalign*}
    \begin{aligned}
    \left(\phi_{ik}-\phi_{jk}\right)\phi_{k}^{-1}+\left(\phi_{ij}-\phi_{ik}\right)\phi_{i}^{-1}+\left(\phi_{jk}-\phi_{ij}\right)\phi_{j}^{-1}=0,\\
    \left(\phi_{ik}^{-1}\left(\chi_{ik}+\chi_k\right)-\phi_{jk}^{-1}\left(\chi_{jk}+\chi_k\right)\right)\phi_k+
    \left(\phi_{ij}^{-1}\left(\chi_{ij}+\chi_i\right)-\phi_{ik}^{-1}\left(\chi_{ik}+\chi_i\right)\right)\phi_i\\
    +\left(\phi_{jk}^{-1}\left(\chi_{jk}+\chi_j\right)-\phi_{ij}^{-1}\left(\chi_{ij}+\chi_j\right)\right)\phi_j=0,
    \end{aligned}&& (\mathcal{X}_4)
    \end{flalign*}
    \begin{flalign*}
    \begin{aligned}
    \left(\phi_{ik}-\phi_{jk}\right)\phi_{k}^{-1}+\left(\phi_{ij}-\phi_{ik}\right)\phi_{i}^{-1}+\left(\phi_{jk}-\phi_{ij}\right)\phi_{j}^{-1}=0,\\
    \left(\phi_{ik}\left(\omega_{ik}-\omega_k\right)-\phi_{jk}\left(\omega_{jk}-\omega_k\right)\right)\phi_k^{-1}+
    \left(\phi_{ij}\left(\omega_{ij}-\omega_i\right)-\phi_{ik}\left(\omega_{ik}-\omega_i\right)\right)\phi_i^{-1}\\
    +\left(\phi_{jk}\left(\omega_{jk}-\omega_j\right)-\phi_{ij}\left(\omega_{ij}-\omega_j\right)\right)\phi_j^{-1}=0.
    \end{aligned}&& (\mathcal{X}^\dag_4)
    \end{flalign*}
\end{prop}

\begin{proof}
The proof of the first two items $(i)$ and $(ii)$ of part $A.$ of the Proposition have been essentially proved in Lemma \ref{1st_lemmam}.

Let us prove item $(iii)$.  Since $T^{i,j}:=v^{i}-v^{j},$ we have
\begin{align*}
T^{i,j}_k+T^{j,k}_i+T^{k,i}_j=\underbrace{v^{i}_k-v^{k}_i}_{=u^k-u^i}+\underbrace{v^{j}_i-v^{i}_j}_{=u^i-u^j}+\underbrace{v^{k}_j-v^{j}_k}_{=u^j-u^k}=0,
\end{align*}
where we have substituted (\ref{eq:2:1}).\\
\noindent $(iv)$ Since $U^{i,j}:=\phi_i^{-1}u^i\phi-\phi_j^{-1}u^j\phi,$ we have
\begin{align*}
U^{i,j}_k+U^{j,k}_i+U^{k,i}_j=\underbrace{\phi_{ij}^{-1}\left(u^j_i\phi_i-u^i_j\phi_j\right)}_{=\left(\phi_j^{-1}u^j-\phi_i^{-1}u^i\right)\phi}
+\underbrace{\phi_{jk}^{-1}\left(u^k_j\phi_j-u^j_k\phi_k\right)}_{=\left(\phi_k^{-1}u^k-\phi_j^{-1}u^j\right)\phi}
+\underbrace{\phi_{ik}^{-1}\left(u^i_k\phi_k-u^k_i\phi_i\right)}_{=\left(\phi_i^{-1}u^i-\phi_k^{-1}u^k\right)\phi}=0,
\end{align*}
where we have substituted (\ref{omega}).\\
\noindent $B.$ Rewriting the closure relations $(i)-(iv)$ in terms of the potentials  (\ref{x}),(\ref{v}), they become exactly the three-dimensional vertex system $(\mathcal{X}_2).$ On the other hand, rewriting them in terms of the potentials (\ref{x}),(\ref{phi}), they coincide with $(\mathcal{X}_4)$. Finally, expressing the closure relations in terms of potentials (\ref{phi}),(\ref{omega}), we obtain $(\mathcal{X}^\dag_4).$
\end{proof}
  We  recall   that the first equation of $(\mathcal{X}_2),$ is related to a non-abelian three-dimensional equation that was introduced in \cite{Nijhoff:1990}, that is 
 \begin{equation}\label{frankx2}
 (\psi_{ik}-\psi_{k})(\psi_{jk}-\psi_{k})^{-1}(\psi_{ij}-\psi_{i})(\psi_{ik}-\psi_{i})^{-1}(\psi_{jk}-\psi_{j})(\psi_{ij}-\psi_{j})^{-1}=1
 \end{equation}   
 via  the B\"acklund transformation
 $\psi_i-\psi=\phi_i\phi^{-1}$. The abelian version of (\ref{frankx2})  coincides with the equation referred to as $(\chi_2)$ in \cite{ABS3d}. Also, the first equation of the three-dimensional vertex systems $(\mathcal{X}_4)$ and $(\mathcal{X}^\dag_4),$ first appeared in \cite{Nijhoff:1990} and in the Abelian limit  coincides with the equation referred to as $(\chi_4)$ in \cite{ABS3d}.

\section{A unified approach on discrete analytic functions} \label{secf}

In this Section we show that both the linear and the nonlinear theories of discrete analytic functions can be considered as special cases of (\ref{ido_i}). We start with the following definition.

\begin{definition}[Discrete analytic functions on an embedding of the ${\mathbb Z}^2$ graph \cite{Mercat77,BiaKasNie}]\label{deeefff}
Let $\sigma: {\mathbb Z}^2 \to {\mathbb C}$ be an injective embedding of the ${\mathbb Z}^2$ graph to the complex plane ${\mathbb C}$ and let the function
$F: {\mathbb C} \to {\mathbb C}$. Then the composition $\chi = F \circ \sigma$ is another embedding of the ${\mathbb Z}^2$ graph to the complex plane ${\mathbb C}$.
We say that the function $F$ is {\em discrete analytic on the embedding $\sigma$ at an elementary quad that consists of the vertices  $\{(m,n), (m+1,n), (m,n+1), (m+1,n+1)\}$ of the graph}, iff the equation
\begin{equation} \label{dbar}
(\sigma_{i}-\sigma_{j})(\chi_{ij}-\chi) =(\sigma_{ij}-\sigma)(\chi_{i}-\chi_{j}),
\end{equation}
holds for the quad. 
We say that the function $F$ is {\em discrete analytic on the embedding $\sigma$}, iff the equation (\ref{dbar})
holds for every elementary quad of the graph.
\end{definition}

\begin{remark}
Since the definition is "local" i.e. it is restricted to elementary quads, it can be extended to any planar quad graph.
\end{remark}
\begin{remark}
Equation (\ref{dbar})  is a discrete analogue of $\bar{\partial} f(z)=0,$ that 
 in discrete integrable systems theory is referred to as
 {\em the discrete Moutard equation} \cite{NimSch,DoGriNieSa}.
\end{remark}

In the literature, {\em discrete analytic functions} usually are defined with a prescribed embedding $\sigma$.
In Ferrand's definition \cite{Ferrand_1944,Duffin1,DuffinDuris,Hayabara,Deeter,ZeilFur,ZeilNew,AlpayJSV} the   embedding $\sigma$ is a square embedding $\left\{ m+in \in {\mathbb C}\, | \, (m,n)\in {\mathbb Z} ^2 \right\}$.
 Duffin's definition \cite{Duffin2, Mercat1} assumes that the embedding consists of rhombi only.
In article \cite{Mercat2} Mercat deals with rectangle embeding $\sigma$.
In this article we assume parallelogram embedding, 
i.e. for every quad of the graph the following constraint holds
\begin{equation} \label{parall}
\sigma_{ij}+\sigma-\sigma_{i}-\sigma_{j}=0.
\end{equation}
However, we would like to stress that  
 in Mercat approach \cite{Mercat77}, where essentialy discrete analytic transformations are
 the transformations that preserves the so called {\em discrete conformal structure}, the role of the 
 embeding $\sigma$ is secondary.

The important observation is that equation  (\ref{dbar})  can be rewritten as
\begin{equation} \label{cit}
(\sigma_{ij}-\sigma_j)  ( \chi_{ij}+ \chi_j)  + (\sigma_j-\sigma) ( \chi_j+ \chi)+(\sigma - \sigma_i)( \chi_{i}+ \chi)  +
 (\sigma_i-\sigma_{ij}) (\chi_{ij}+ \chi_i) =0.
\end{equation}
If we define the  directed integral over the directed edge  $(a,b)$ of the graph as
\[\int_{(a,b)} \chi(\sigma) d \sigma := \frac{ \chi(b)+ \chi(a)}{2}(\sigma(b)-\sigma(a)),\]
then equation (\ref{dbar}) serves as the
discrete analogue of the Cauchy integral theorem %$\ointclockwise f(z) dz=0$
$\ointctrclockwise  f(z) dz=0$  \cite{Duffin2},
which can be written as
\[\ointctrclockwise _{\Diamond} \chi(\sigma) d \sigma=0.\]
In the formula above $\Diamond$ means that the integration is performed over an elementary quadrilateral of the ${\mathbb Z}^2$ graph.

An important characterisation of the discrete Moutard equation is that it leads to the discrete Laplace equation by the sublattice approach
(see \cite{DoGriNieSa})

\begin{align}
\begin{aligned} \label{Lap} (\sigma_{ij}-\sigma)^{-1} (\sigma_{i}-\sigma_{j})(\chi_{ij}-\chi)+
(\sigma_{-i-j}-\sigma)^{-1}(\sigma_{-i}-\sigma_{-j})(\chi_{-i-j}-\chi)+\\
+(\sigma_{-ij}-\sigma)^{-1}(\sigma_{j}-\sigma_{-i})(\chi_{-ij}-\chi)
+(\sigma_{i-j}-\sigma)^{-1}(\sigma_{-j}-\sigma_{i})(\chi_{i-j}-\chi)=0,\
\end{aligned}
\end{align}
that serves as the discrete analogue of 
$ (\partial _x^2+\partial_y^2) f(x+iy)=0.$

\begin{remark}
The discrete Cauchy integral theorem (\ref{cit}) (or equivalently (\ref{dbar})), is a local property (it is defined on a single quad)
 and  can thus be easily extended to an arbitrary quad-graph. In addition,  the transition from (\ref{dbar}) to the 
 Laplace type equation (\ref{Lap}) 
can be performed in an arbitrary quad-graph.
\end{remark}

Finally, we would like to  stress that the formulas
(\ref{dbar}), (\ref{cit}) and  (\ref{Lap}) are written in a form that is valid in the non-abelian case.

\subsection{Non-abelian unification of the theories of discrete analytic functions.}
 We will show now that both the linear theory (see Definition \ref{deeefff}) and the non-linear theory (see \cite{BoSeSp}) of discrete analytic functions have a unified description. 
 
 The system of difference equations (\ref{top}) under the assumptions:
  \begin{enumerate}
\item
 the values of  functions $\chi$ and $\omega$ commute with the values of any of the functions $\phi$, $\psi$ or $\sigma$;
 \item the functions $\chi$, $\omega$, $\phi$, $\psi$ and $\sigma$  are embeddings  of the ${\mathbb Z}^2$ graph in an associative algebra ${\mathcal A}$ 
i.e. they are considered as  maps ${\mathbb Z}^2 \to {\mathcal A}$;
\item
the values of the functions
$\phi$, $\psi$, $\sigma$ belong to  the domain
and values of functions
$\chi$, $\omega$ belong to the codomain of the maps 
$F_i:{\mathcal A} \to {\mathcal A}$, such that  $\chi =F_1 \circ 
\phi$, $\chi=F_2 \circ \psi$, $\omega=F_3 \circ \sigma$, $\omega=F_4 \circ \phi$, $\omega=F_5 \circ \psi$ and $\omega=F_6 \circ \sigma\footnote{We recall that $\chi =F_1 \circ 
\phi$ means that for every $(n_1,n_2) \in {\mathbb Z}^2$ we have  $\phi(n_1,n_2)\mapsto F_1 \circ \phi(n_1,n_2)=\chi(n_1,n_2) $ and similarly for the remaining functions}$;
\end{enumerate}
 reads
\begin{subequations} \label{analy}
\begin{gather} \label{eqpaa}
 \chi_i+ \chi= \phi_i \phi^{-1}(  \omega_i - \omega),\\ \label{eqpab}
 \chi_i+ \chi= (\psi_i-\psi)(  \omega_i - \omega),\\ \label{eqpac}
 (\sigma_i-\sigma)(\chi_i+ \chi)= (  \omega_i - \omega).
\end{gather}
\end{subequations}

Furthermore, if we assume that if $ \omega_i - \omega \neq 0$ then from equations (\ref{eqpaa}), (\ref{eqpab}) and (\ref{eqpac})
we infer that
\begin{equation} \label{eqpba}
\psi_i-\psi=\phi_i \phi^{-1},
\end{equation}
\begin{equation} \label{eqpbb}
\sigma_i-\sigma=\phi_i^{-1} \phi,
\end{equation}
and the compatibility of (\ref{eqpba}) and (\ref{eqpbb}) leads to the following difference equations with one dependent
variable each:
\begin{equation} \label{sigm}
\phi_{ij} (\phi_j^{-1}-\phi_i^{-1})=(\phi_i-\phi_j) \phi^{-1},
\end{equation}
\begin{equation} \label{siga}
(\psi_{ij}-\psi_j)(\psi_j-\psi)=(\psi_{ij}-\psi_i)(\psi_i-\psi)
\end{equation}
\begin{equation} \label{sigs}
(\sigma_{ij}-\sigma_j)(\sigma_j-\sigma)=(\sigma_{ij}-\sigma_i)(\sigma_i-\sigma)
\end{equation}
In addition  the functions $\mu^i$ and $v^i$
\begin{subequations}\label{eqv}
\begin{equation} \label{eqv1}
\mu^i:=\psi_i-\psi=\phi_i \phi^{-1},
\end{equation}
\begin{equation} \label{eqv2}
v^i:= \chi_i+ \chi= \phi_i \phi^{-1}(  \omega_i - \omega) ,
\end{equation}
\end{subequations}
 are naturally defined on the edges of the ${\mathbb Z}^2$ graph.

Summarizing, both  the linear theory of analytic functions and the non-linear theory  arise as 
special cases of system (\ref{analy}), or more precisely as
special cases of the  structure  $\mathcal{O}(F,\sigma),$ where
\begin{enumerate}[label=(\roman*)]
\item the function $\sigma:{\mathbb Z}^2\rightarrow \mathcal{A}$   is an embedding of the ${\mathbb Z}^2$ graph into an algebra $\mathcal A$;
\item  $F: {\mathcal A} \supset \hbox{Im}\, \sigma \ni x \mapsto y=F(x) \in {\mathbb C},  $  is a function and $\chi:=F \circ \sigma $;
\item the embeddings $\sigma$ and $\chi$ commute i.e. if $x\in \hbox{Im}\, \sigma$ and $y\in \hbox{Im}\, \chi$ then  $x y=y x$;
\item the embeddings satisfy
\begin{equation} \label{dbars}
(\sigma_{i}-\sigma_{j})(\chi_{ij}-\chi) =(\sigma_{ij}-\sigma)(\chi_{i}-\chi_{j}),
\end{equation}
 that is equivalent to
\begin{equation} \label{cits}
(\sigma_{ij}-\sigma_j)  ( \chi_{ij}+ \chi_j)  + (\sigma_j-\sigma) ( \chi_j+ \chi)+(\sigma - \sigma_i)( \chi_{i}+ \chi)  +
 (\sigma_i-\sigma_{ij}) (\chi_{ij}+ \chi_i) =0,
\end{equation}
that implies 
\begin{align}
\begin{aligned} \label{Laps} (\sigma_{ij}-\sigma)^{-1} (\sigma_{i}-\sigma_{j})(\chi_{ij}-\chi)+
(\sigma_{-i-j}-\sigma)^{-1}(\sigma_{-i}-\sigma_{-j})(\chi_{-i-j}-\chi)+\\
+(\sigma_{-ij}-\sigma)^{-1}(\sigma_{j}-\sigma_{-i})(\chi_{-ij}-\chi)
+(\sigma_{i-j}-\sigma)^{-1}(\sigma_{-j}-\sigma_{i})(\chi_{i-j}-\chi)=0;
\end{aligned}
\end{align}
\item  in general no restriction on the embedding $\sigma$ is imposed.
\end{enumerate}

\subsubsection{Linear theory of discrete analytic functions}

If one considers ${\mathcal A}= {\mathbb C}$, the  system of equations (\ref{eqpac}) and (\ref{sigs})  coincides with  the Definition \ref{deeefff} of  discrete analytic functions on parallelogram embedings.
Indeed,  equation (\ref{eqpac}) after elimination $\omega$ gives  (\ref{dbar}), whereas equation (\ref{sigs})
in the abelian case
can be rewritten as
$(\sigma_{ij}-\sigma_i-\sigma_j+\sigma)(\sigma_{i}-\sigma_j)=0,$ which means that
 either the faces of the domain are parallelograms or they locally degenerate to a line.
 
 We plan to dedicate a separate article to investigating further consequences of system  (\ref{analy})
in the theory of discrete analytic functions.

\subsubsection{Non-linear theory of analytic functions}

We show now that the non-linear approach of discrete analytic functions theory (see e.g. \cite{Thurston,Stephenson:2005,Bobe_2006,Hoffmann:2003,BoSeSp,Joshi_2021})  is a special
case of the non-abelian theory. 

We follow the steps we have used in Lemma \ref{lema000}.
We take
\begin{align*}
\mu^i:=\psi_i-\psi=\begin{bmatrix} 0  & a^i \\ b^i & 0 \end{bmatrix}, \qquad v^i:= \begin{bmatrix} \chi_i+ \chi  & 0 \\  & \chi_i+ \chi \end{bmatrix}.
\end{align*}
where the entries of the matrices are functions $\mathbb{Z}^2\to \mathbb{C}$.
Equations (\ref{eqq:2:2}), (\ref{eqq:2:3})  can be rewritten as
\[ a^i_j =\frac{a^i-a^j}{b^i-b^j} b^i, \quad b^i_j =\frac{b^i-b^j}{a^i-a^j} a^i.\]
Multiplying the sides of the equations we get
\[ (a^i b^i)_j =a^i b^i, \]
so the product of  the two variables $(\alpha^i)^2:=a^i b^i$ is a function of the $i$-th independent variable only.
 So we can set 
 \[ a^i= \alpha^i y^i, \quad b^i =\frac{\alpha^i}{y^i}.\]
We conclude that $\mu^i$ should be of the form
\begin{align*}
\mu^i:=\psi_i-\psi=\begin{bmatrix} 0  & \alpha^i y^i \\ \frac{\alpha^i}{y^i} & 0 \end{bmatrix},
\end{align*}
 where $\alpha^i$ are given $\mathbb C$ valued functions which depend only on $i$-th independent variable of the graph,
and $y^i$ 
 are  $\mathbb C$ valued functions that have to be determined.
It turns out that the functions $y^i$ must obey
\begin{subequations}\label{eqn}
\begin{equation} \label{eqn1}
y^i_jy^i =y^j_iy^j,
\end{equation}
\begin{equation} \label{eqn2}
\alpha^i (y^i_j - y^i) =\alpha^j (y^j_i - y^j).
\end{equation}
\end{subequations}
Equation (\ref{eqn1}) guarantees the existence of potential $\w$ such that
\begin{align*}
y^i=\w_i \w,
\end{align*}
while  equation (\ref{eqn2})  becomes the discrete modified KdV equation \cite{NiRaGrOh}
\begin{align*}
\alpha^i (\w_{ij} \w_j - \w_i\w) =&\alpha^j (\w_{ij} \w_i- \w_j \w).
\end{align*}
Then equation (\ref{eqpaa}) reads
\begin{align*}
(\chi_i+\chi) =& \w_i \w (\omega_i-\omega).
\end{align*}
Applying the point transformation
\begin{align*}
\chi \to (-1)^{n_1+n_2} \chi, && y^i \to -(-1)^{n_1+n_2} y^i,
\end{align*}
we obtain  the system of equations
\begin{align*}
(\chi_i-\chi) =& \w_i \w (\omega_i-\omega),\\
\alpha^i (\w_{ij} \w_j + \w_i\w) =&\alpha^j (\w_{ij} \w_i+ \w_j \w),
\end{align*}
which constitutes the foundations of the alternative non-linear
integrable  discretization of analytic functions
\cite{BoSeSp}.

\section*{Acknowledgements}
\parbox{.135\textwidth}{\begin{tikzpicture}[scale=.03]
\fill[fill={rgb,255:red,0;green,51;blue,153}] (-27,-18) rectangle (27,18);
\pgfmathsetmacro\inr{tan(36)/cos(18)}
\foreach \i in {0,1,...,11} {
\begin{scope}[shift={(30*\i:12)}]
\fill[fill={rgb,255:red,255;green,204;blue,0}] (90:2)
\foreach \x in {0,1,...,4} { -- (90+72*\x:2) -- (126+72*\x:\inr) };
\end{scope}}
\end{tikzpicture}} \parbox{.85\textwidth}
{This research is part of the project No. 2022/45/P/ST1/03998  co-funded by the National Science Centre and the European Union Framework Programme
 for Research and Innovation Horizon 2020 under the Marie Sklodowska-Curie grant agreement No. 945339. For the purpose of Open Access, the author has applied a CC-BY public copyright licence to any Author Accepted Manuscript (AAM) version arising from this submission.}

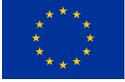
\appendix

\section{Proof of Lemma \ref{1st_lemmam}}\label{app0}
We prove the three items of  lemma \ref{1st_lemmam}.
\begin{enumerate}
\item Since $K^{i,j}:=\m^{i}-\m^{j},$ we have
\begin{align*}
K^{i,j}_k+K^{j,k}_i+K^{k,i}_j=\m^{i}_k-\m^{j}_k+\m^{j}_i-\m^{k}_i+\m^{k}_j-\m^{i}_j=0,
\end{align*}
where we have substituted the expressions $\m^{p}_q,$ $p\neq q\in \{i,j,k\}$ from (\ref{CM_MAPS_12}).
\item Applying the relations (\ref{CM_MAPS_12}), we verify the second  point of the lemma.
\item We now prove the third  point.  There is
\begin{multline*}
  {}^i\Gamma^{j,k}=K^{i,j}_k K^{i,k}=\left(\m^{i}_k-\m^{j}_k\right)K^{i,k}=
  \left(K^{i,k}\m^{i}\left(K^{i,k}\right)^{-1}-K^{j,k}\m^{j}\left(K^{j,k}\right)^{-1}\right)K^{i,k}\\
  =K^{j,k}\left(\underbrace{\left(K^{j,k}\right)^{-1}K^{i,k}\m^{i}-\m^{j}\left(K^{j,k}\right)^{-1}K^{i,k}}_{{}^i\Delta^{j,k}}\right)  .
\end{multline*}
The expressions ${}^i\Delta^{j,k}$ are antisymmetric with respect to the interchange $j\leftrightarrow k.$ %that is $\Delta^{jk}+\Delta^{kj}=0.$
Indeed, under expansion and recollection of terms ${}^i\Delta^{j,k}$ reads
\begin{multline*}
{}^i\Delta^{j,k}=\left(\m^{j}-\m^{k}\right)^{-1}\left(\m^{i}\right)^2+\left(\left(\m^{k}\right)^{-1}-\left(\m^{j}\right)^{-1}\right)^{-1}\\
             -\left(\underbrace{\left(1-\m^{k}\left(\m^{j}\right)^{-1} \right)^{-1}-\left(1-\left(\m^{k}\right)^{-1}\m^{j} \right)^{-1}}_{E^{j,k}}\right)\m^{i}
\end{multline*}
The first two terms in the expression above are clearly antisymmetric under the interchange $j$ to $k.$ The same is true for the term  $E^{j,k}$ since we have
  \begin{multline*}
E^{j,k}=\left(1-\m^{k}\left(\m^{j}\right)^{-1} \right)^{-1}-\left(1-\left(\m^{k}\right)^{-1}\m^{j} \right)^{-1}\\
                = 1-\left(1-\m^{j}\left(\m^{k}\right)^{-1} \right)^{-1}-\left( 1- \left(1-\left(\m^{j}\right)^{-1}\m^{k} \right)^{-1}\right)=-E^{k,j},
\end{multline*}
where we have used twice the identity $\left(1-AB^{-1}\right)^{-1}+\left(1-BA^{-1}\right)^{-1}=1$. To recapitulate, since    both $K^{j,k}$ and  $ {}^i\Delta^{j,k},$ are  antisymmetric expressions under the interchange $j$ to $k,$ their product that coincides with ${}^i\Gamma^{j,k}$ $({}^i\Gamma^{j,k}=K^{j,k} {}^i\Delta^{j,k}),$ is symmetric.

Let us now prove relations (\ref{lemma_31}).  We have
\begin{multline*}
{}^i\Gamma^{j,k}\m^{i}\left({}^i\Gamma^{j,k}\right)^{-1} {}^i\Psi^{j,k}=\\
{}^i\Gamma^{j,k}\m^{i}\left(\underline{\left({}^i\Gamma^{j,k}\right)^{-1}}\left(1-\m^{j}_k\left(\m^{i}_k\right)^{-1}\right)\left(v^{k}-v^{i}\right)-
\underline{\left({}^i\Gamma^{j,k}\right)^{-1}}\left(1-\m^{k}_j\left(\m^{i}_j\right)^{-1}\right)\left(v^{j}-v^{i}\right)\right)\\
=\underline{{}^i\Gamma^{j,k}} \m^{i}\left(K^{i,k}\right)^{-1}\left(\m^{i}_k\right)^{-1}\left(v^{k}-v^{i}\right)-
\underline{{}^i\Gamma^{k,j}} \m^{i}\left(K^{i,j}\right)^{-1}\left(\m^{i}_j\right)^{-1}\left(v^{j}-v^{i}\right)\\
=\left(\m^{i}_k-\m^{j}_k\right)\left(v^{k}-v^{i}\right)-\left(\m^{i}_j-\m^{k}_j\right)\left(v^{j}-v^{i}\right)=-{}^i\Omega^{j,k}.
\end{multline*}
 where we have used  (\ref{def_g})  that results
\begin{align}
  {}^i\Gamma^{j,k}=\left(1-\m^{j}_k\left(\m^{i}_k\right)^{-1}\right)\m^{i}_k K^{i,k}= {}^i\Gamma^{k,j}=\left(1-\m^{k}_j\left(\m^{i}_j\right)^{-1}\right)\m^{i}_j K^{i,j},
\end{align}
to eliminate the underlined terms of the latter expressions.
\end{enumerate}

\section{Proof of Theorem \ref{theo1m}} \label{app00}
In order to prove that mappings $Q_{ij}$ are multidimensional compatible we have to prove that
$
v^{i}_{jk}=v^{i}_{kj},$ $m^{i}_{jk}=m^{i}_{kj},$ $\forall i\neq j\neq k\neq i\in \{1,\ldots,N\}$.

Let us first prove $m^{i}_{jk}=m^{i}_{kj},$ or equivalently $\m^{i}_{jk}=\m^{i}_{kj},$ where $\m^{i}:=\left(m^{i}\right)^{-1}$. Shifting  (\ref{CM_MAPS_11_n})
%, that is $\m^{i}_j=\left(\m^{i}-\m^{j}\right)\m^{i}\left(\m^{i}-\m^{j}\right)^{-1}$
        to the $k-$direction, we have
\begin{align}\label{proof1:1}
\m^{i}_{jk}=\left(\m^{i}_k-\m^{j}_k\right)\m^{i}_k\left(\m^{i}_k-\m^{j}_k\right)^{-1}.
\end{align}
There is
%\begin{align}\label{proof1:2}
\begin{multline*}
\m^{i}_k-\m^{j}_k=\left(\m^{i}-\m^{k}\right)\m^{i}\left(\m^{i}-\m^{k}\right)^{-1}-
\left(\m^{j}-\m^{k}\right)\m^{j}\left(\m^{j}-\m^{k}\right)^{-1}
={}^i\Gamma^{j,k}\left(K^{i,k}\right)^{-1},
\end{multline*}
%\end{align}
so (\ref{proof1:1}) reads
\begin{align*}
\m^{i}_{jk}={}^i\Gamma^{j,k}\left(K^{i,k}\right)^{-1}\m^{i}_k K^{i,k}\left({}^i\Gamma^{j,k}\right)^{-1}
\end{align*}

Substituting  $\m^{i}_{k}$ that reads $\m^{i}_{k}=K^{i,k}\m^{i}\left(K^{i,k}\right)^{-1}$  into the relations above we obtain the multidimensional compatibility formula
\begin{align}\label{ccf1}
\m^{i}_{jk}={}^i\Gamma^{j,k}\m^{i}\left({}^i\Gamma^{j,k}\right)^{-1},
\end{align}
that is clearly symmetric under the interchange $j$ to $k$ since  $\Gamma^{j,k}$  is symmetric (see Lemma \ref{1st_lemmam}) under the same interchange. Finaly, it can be shown easily that
\begin{align}\label{cc_f_m}
\m^{i}_{jk}=\m^{j}_{ik}+K^{i,j}_k.
\end{align}

Now we prove that
$
v^{i}_{jk}-v^{i}_{kj}=0,$  $\forall i\neq j\neq k\neq i\in \{1,\ldots,N\}$. Shifting (\ref{CM_MAPS_12_n}) in the $k-th$ direction we have
\begin{align*}
v^{i}_{jk}K^{i,j}_k=\m^{i}_{jk}v^{j}_k-\m^{j}_{ik}v^{i}_k+\left(v^{j}_k-v^{i}_k\right)\m^{i}_k.
\end{align*}
Using (\ref{def_g}),(\ref{cc_f_m}) and $\m^{i}_k=K^{i,k}\m^{i}\left(K^{i,k}\right)^{-1},$ the relations above read
\begin{align*}
v^{i}_{jk}{}^i\Gamma^{j,k}=\m^{i}_{jk}(v^{j}_k-v^{i}_k)K^{i,k}+{}^i\Gamma^{i,j}\left(K^{i,k}\right)^{-1}v^{i}_kK^{i,k}
+\left(v^{j}_k-v^{i}_k\right)K^{i,k}\m^{i}.
\end{align*}
Since ${}^i\Gamma^{j,k}$ is symmetric under $j\leftrightarrow k,$ we have
\begin{align}\label{tp1}
\left(v^{i}_{jk}-v^{i}_{kj}\right){}^i\Gamma^{j,k}=\m^{i}_{jk} {}^i\Omega^{j,k}+{}^i\Gamma^{j,k}
\left(\left(K^{i,k}\right)^{-1}v^{i}_kK^{i,k}-\left(K^{i,j}\right)^{-1}v^{i}_jK^{i,j}\right)+{}^i\Omega^{j,k}\m^{i},
\end{align}
where ${}^i\Omega^{j,k}:=\left(v^{j}_k-v^{i}_k\right)K^{i,k}-\left(v^{k}_j-v^{i}_j\right)K^{i,j}.$ Using (\ref{CM_MAPS_12_n}) shifted accordingly, we substitute $v^{j}_k,v^{i}_k,v^{k}_j$ and $v^{i}_j$ into ${}^i\Omega^{j,k}$ and the latter coincides with the expressions (\ref{ksi}) of Lemma \ref{1st_lemmam}.
%%%%%%%%%%%%
  In order  to prove the multidimensional compatibility we have to prove that the Rhs of (\ref{tp1}) identically vanishes. %Indeed, we take the first two %terms of  Rhs of (\ref{tp1}) and by after using (\ref{ccf1})  we expand and  recollect terms to obtain respectively
  Indeed, after using (\ref{ccf1}) the Rhs of (\ref{tp1}) reads
  \begin{align}\label{les1}
  \begin{aligned}
{}^i\Gamma^{j,k}\m^{i}\left({}^i\Gamma^{j,k}\right)^{-1}\left(\underbrace{{}^i\Omega^{j,k}+{}^i\Gamma^{j,k}\left(\m^{i}\right)^{-1}
\left(\left(K^{i,k}\right)^{-1}v^{i}_kK^{i,k}-\left(K^{i,j}\right)^{-1}v^{i}_jK^{i,j}\right)}_{{}^i\Lambda^{j,k}}\right)\\
+{}^i\Omega^{j,k}\m^{i}.
\end{aligned}
  \end{align}
  We use (\ref{CM_MAPS_12_n}) shifted accordingly and after expansion and recollection of terms, the expressions  ${}^i\Lambda^{j,k}$ reads
  \begin{multline*}
  {}^i\Lambda^{j,k}=\underbrace{\left(\m^{j}_k\left(\m^{i}_k\right)^{-1}\m^{k}_i-\m^{k}_j\left(\m^{i}_j\right)^{-1}\m^{j}_i\right)}_{=0\;\; \mbox{due to Lemma \ref{1st_lemmam} }}v^{i}\\
  +\left(\left(1-\m^{j}_k\left(\m^{i}_k\right)^{-1}\right)\left(v^{k}-
v^{i}\right)-\left(1-\m^{k}_j\left(\m^{i}_j\right)^{-1}\right)\left(v^{j}-v^{i}\right)\right)\m^{i}%={}^i\Psi^{j,k}\m^{i},
  \end{multline*}
  or (see (\ref{psi}))
  \begin{align}\label{lambda}
  {}^i\Lambda^{j,k}={}^i\Psi^{j,k}\m^{i}.
  \end{align}
  Using (\ref{lambda}),  the expressions (\ref{les1}), that constitute the Rhs of (\ref{tp1}) read
\begin{align*}
{}^i\Gamma^{j,k}\m^{i}\left({}^i\Gamma^{j,k}\right)^{-1}{}^i\Lambda^{j,k}+{}^i\Omega^{j,k}\m^{i}=
\underbrace{\left({}^i\Gamma^{j,k}\m^{i}\left({}^i\Gamma^{j,k}\right)^{-1}\Psi^{j,k}+\Omega^{j,k}\right)}_{=0\;\; \mbox{due to Lemma \ref{1st_lemmam} }}\m^{i}=0
\end{align*}
 and that completes the proof.

%\bibliographystyle{unsrt}
%\bibliography{ref_list3}

\end{document}